\documentclass[a4paper, UKenglish, cleveref, autoref, thm-restate, ]{lipics-v2021}

\usepackage[utf8]{inputenc}
\usepackage{amsmath,amssymb,amsthm}
\usepackage[textwidth=2cm,textsize=small]{todonotes}
\usepackage{stmaryrd}
\usepackage{microtype}
\usepackage{url}
\usepackage{doi}
\usepackage{hyperref}
\usepackage{cleveref}
\usepackage{xspace}
\usepackage{paralist}
\usepackage{appendix}

\NewDocumentEnvironment{code}{}
 {\begin{minted}{haskell}}
 {\end{minted}}
 
\title{Bidimensional linear recursive sequences and universality of unambiguous register automata}
\author{Corentin Barloy}{\'Ecole Normale Supérieure de Paris, PSL, France}{corentin.barloy@ens.fr}{}{Partially supported by the Polish NCN grant 2017/26/D/ST6/00201.}
\author{Lorenzo Clemente}{University of Warsaw, Poland}{clementelorenzo@gmail.com}{https://orcid.org/0000-0003-0578-9103}{Partially supported by the Polish NCN grant 2017/26/D/ST6/00201.}

\authorrunning{Corentin Barloy and Lorenzo Clemente}

\titlerunning{Linear recursive sequences and universality of register automata}

\Copyright{Corentin Barloy and Lorenzo Clemente}

\date{January 2021}



\acknowledgements{We would like to thank Daniel Robertz for kindly providing us with the \texttt{LDA} package for Maple 16.}

\nolinenumbers

\EventEditors{John Q. Open and Joan R. Access}
\EventNoEds{2}
\EventLongTitle{42nd Conference on Very Important Topics (CVIT 2016)}
\EventShortTitle{CVIT 2016}
\EventAcronym{CVIT}
\EventYear{2016}
\EventDate{December 24--27, 2016}
\EventLocation{Little Whinging, United Kingdom}
\EventLogo{}
\SeriesVolume{42}
\ArticleNo{23}

\newcommand{\A}{\mathbb A}

\newcommand{\F}{\mathbb F}

\newcommand{\N}{\mathbb N}
\newcommand{\Z}{\mathbb Z}
\newcommand{\Q}{\mathbb Q}

\renewcommand{\AA}{A}
\newcommand{\BB}{B}
\newcommand{\CC}{C}

\newcommand{\shift}[1]{\partial_{#1}}

\newcommand{\tuple}[1]{(#1)}
\newcommand{\set}[1]{\{#1\}}
\newcommand{\setof}[2]{\set{#1 \mid #2}}
\newcommand{\orbit}[1]{[#1]}
\newcommand{\orbits}[2]{\mathsf{orbits}_{#1}(#2)}
\newcommand{\card}[1]{\left|#1\right|}
\newcommand{\goesto}[1]{\xrightarrow{#1}}
\newcommand{\lang}[1]{L(#1)}
\newcommand{\blang}[1]{L^{\mathsf R}(#1)}
\newcommand{\width}[1]{\mathsf{\#}#1}

\newcommand{\sem}[1]{\left\llbracket#1\right\rrbracket}

\newcommand{\RNR}[1]{G_{#1}}
\newcommand{\runs}[3]{\mathsf{Runs}(#1; #2; #3)}
\newcommand{\stirling}[2]{S(#1, #2)}

\newcommand\NCTWO{\textsf{\textup{NC}}${}^2$\xspace}
\newcommand\NLOGSPACE{\textsf{\textup{NL}}\xspace}
\newcommand\PTIME{\textsf{\textup{PTIME}}\xspace}
\newcommand\PSPACE{\textsf{\textup{PSPACE}}\xspace}
\newcommand\EXPTIME{\textsf{\textup{EXPTIME}}\xspace}
\newcommand\EXPSPACE{\textsf{\textup{EXPSPACE}}\xspace}
\newcommand\TWOEXPTIME{\textsf{\textup{2-EXPTIME}}\xspace}
\newcommand\TWOEXPSPACE{\textsf{\textup{2-EXPSPACE}}\xspace}

\newcommand\HNF{\textsf{HNF}\xspace}

\newcommand\CR{\textsf{CR}\xspace}
\newcommand{\softO}[1]{\tilde{O}({#1})}

\newcommand{\size}[1]{\abs{#1}_\infty}
\newcommand{\abs}[1]{\left|{#1}\right|}

\newcommand{\I}{\textsf{\bf I}}
\newcommand{\II}{\textsf{\bf II}}
\newcommand{\III}{\textsf{\bf III}}

\renewcommand{\L}{\mathsf L}

\newcommand{\CLM}{\textsf{CLM}\xspace}

\newcommand{\decision}[3]{\medskip\noindent {\sc #1}. \newline {\bf Input: } #2 \newline {\bf Output: } #3\medskip}

\newif\ifstartedinmathmode
\renewcommand*{\st}{
  \relax\ifmmode\startedinmathmodetrue\else\startedinmathmodefalse\fi
  \ifstartedinmathmode{\;\cdot\;}\else{s.t.~}\fi%
}

\makeatletter
\def\cleartheorem#1{%
    \expandafter\let\csname#1\endcsname\relax
    \expandafter\let\csname c@#1\endcsname\relax
}
\makeatother

\crefname{equation}{}{}
\crefname{claim}{Claim}{Claims}
\crefname{section}{Sec.~}{Sections}

\begin{document}

\maketitle

\begin{abstract}
    We study the universality and inclusion problems for register automata over equality data $(\A, =)$.
    We show that the universality $\lang B = (\Sigma \times \A)^*$
    and inclusion problems $\lang A \subseteq \lang B$
    can be solved with \TWOEXPTIME complexity
    when both automata are without guessing and $B$ is unambiguous,
    improving on the currently best-known \TWOEXPSPACE upper bound by Mottet and Quaas.
    When the number of registers of both automata is fixed,
    we obtain a lower \EXPTIME complexity,
    also improving the \EXPSPACE upper bound from Mottet and Quaas
    for fixed number of registers.
    %
    %
    We reduce inclusion to universality,
    and then we reduce universality to the problem of counting the number of orbits of runs of the automaton.
    We show that the orbit-counting function satisfies a system of bidimensional linear recursive equations with polynomial coefficients (linrec),
    which generalises analogous recurrences for the Stirling numbers of the second kind,
    and then we show that universality reduces to the zeroness problem for linrec sequences.
    While such a counting approach is classical
    and has successfully been applied to unambiguous finite automata and grammars over finite alphabets,
    its application to register automata over infinite alphabets is novel.
    
    We provide two algorithms to decide the zeroness problem
    for bidimensional linear recursive sequences
    arising from orbit-counting functions.
    Both algorithms rely on techniques from linear non-commutative algebra.
    The first algorithm performs variable elimination and has elementary complexity.
    The second algorithm is a refined version of the first one
    and it relies on the computation of the Hermite normal form of matrices over a skew polynomial field.
    The second algorithm yields an \EXPTIME decision procedure for the zeroness problem of linrec sequences, which in turn yields the claimed bounds for the universality and inclusion problems of register automata.
\end{abstract}

\newpage

\section{Introduction}

\subparagraph{Register automata.}

\emph{Register automata} extend finite automata with finitely many registers
holding values from an infinite \emph{data domain $\A$}
which can be compared against the data appearing in the input.
The study of register automata arises naturally in automata theory
as a conservative generalisation of finite automata over finite alphabets $\Sigma$
to richer but well-behaved classes of infinite alphabets.
The seminal work of Kaminski and Francez
introduced \emph{finite-memory automata} 
as the study of register automata over the data domain $(\A, =)$
consisting of an infinite set $\A$ and the equality relation
\cite{KaminskiFrancez:TCS:1994}.
The recent book \cite{Bojanczyk:AtomBook:2019}
studies automata theory over other data domains
such as $(\Q, \leq)$,
and more generally homogeneous \cite{Macpherson:Survey:DM:2011}
or even $\omega$-categorical relational structures.
Another motivation for the study of register automata comes from the area of database theory:
XML documents can naturally be modelled as finite unranked trees
where data values from an infinite alphabet are necessary to model the \emph{attribute values} of the document
(c.f.~\cite{NevenSchwentickVianu:TOCL:2004} and the survey \cite{Segoufin:CSL06}).


The central verification question for register automata
is the \emph{inclusion problem},
which, for two given automata $A, B$,
asks whether $\lang A \subseteq \lang B$.
In full generality the problem is undecidable
and this holds already in the special case of the \emph{universality problem} $\lang B = (\Sigma \times \A)^*$ \cite[Theorem 5.1]{NevenSchwentickVianu:TOCL:2004},
when $B$ has only two registers \cite[Theorem 1.8]{Bojanczyk:AtomBook:2019}
(or even just one register in the more powerful model \emph{with guessing} \cite[Exercise 9]{Bojanczyk:AtomBook:2019},
i.e., non-deterministic reassignment in the terminology of \cite{KaminskiZeitlin:IJFCS:2010}).
One way to obtain decidability is to restrict the automaton $B$.
One such restriction requires that $B$ is \emph{deterministic}:
Since deterministic register automata are effectively closed under complementation,
the inclusion problem reduces to non-emptiness of $\lang A \cap (\Sigma \times \A)^* \setminus \lang B$,
which can be checked in \PSPACE.
Another, incomparable, restriction demands
that $B$ has only one register:
In this case the problem becomes decidable
\cite[Appendix A]{KaminskiFrancez:TCS:1994}%
\footnote{Decidability even holds for the so-called ``two-window register automata'',
which combined with the restriction in \cite{KaminskiFrancez:TCS:1994}
demanding that the last data value read must always be stored in some register
boils down to a slightly more general class of ``$1\frac 1 2$-register automata''.}
and non-primitive recursive \cite[Theorem 5.2]{DemriLazic:FreezeLTL:ACM09}.

\subparagraph{Unambiguity.}

\emph{Unambiguous automata} are a natural class of automata intermediate between deterministic and nondeterministic automata.
An automaton is unambiguous if
there is at most one accepting run on every input word.
Unambiguity has often been used to generalise decidability results for deterministic automata
at the price of a usually modest additional complexity.
For instance, the universality problem for deterministic finite automata
(which is \PSPACE-complete in general \cite{StockmeyerMeyer:1973})
is \NLOGSPACE-complete,
while for the unambiguous variant it is in \PTIME \cite[Corollary 4.7]{StearnsHunt:Unambiguous:SFCS81},
and even in \NCTWO \cite{Tzeng:IPL:1996}.
An even more dramatic example is provided by universality of context-free grammars,
which is undecidable in general \cite[Theorem 9.22]{HopcroftMotwaniUllman:2000},
\PTIME-complete for deterministic context-free grammars,
and decidable for unambiguous context-free grammars \cite[Theorem 5.5]{SalomaaSoittola:Book:PowerSeries:1978}
(even in \PSPACE \cite[Theorem 10]{Clemente:EPTCS:2020}).
(The more general equivalence problem is decidable for deterministic context-free grammars \cite{Senizergues:ICALP:1997},
but it is currently an open problem whether equivalence is decidable for unambiguous context-free grammars,
as well as for the more general \emph{multiplicity equivalence} of context-free grammars \cite{Kuich:multiplicity:1994}.)
Other applications of unambiguity for universality and inclusion problems in automata theory include
Büchi automata \cite{BousquetLoeding:Unambiguous:LATA:2010,BaierKieferKleinKluppelholzMullerWorrell:CAV:2016},
probabilistic automata \cite{DaviaudJurdzinskiLazicMazowieckiPerezWorrell:ICALP:2018},
Parikh automata \cite{CadilhacFinkelMcKenzie:2012,BostanCarayolKoechlinNicaud:ICALP:2020},
vector addition systems \cite{CzerwinskiFigueiraHofman:CONCUR:2020},
and several others (c.f.~also \cite{Colcombet:STACS:2012,Colcombet:DCFS:2015}).

\subparagraph{Number sequences and the counting approach.}

The universality problem for a language over finite words $L \subseteq \Sigma^*$
is equivalent to whether its associated \emph{word counting function}
$f_L(n) := \card{L \cap \Sigma^n}$
equals $\card\Sigma^n$ for every $n$.
The most classical way of exploiting unambiguity of a computation model $A$
(finite automaton, context-free grammar, \ldots)
is to use the fact that it yields a bijection between the recognised language $\lang A$
and the set of accepting runs.
In this way, $f_L(n)$ is also the number of accepting runs of length $n$,
and for the latter recursive descriptions usually exist.
When the class of number sequences to which $f_L$ belongs contains $\card\Sigma^n$
and is closed under difference,
this is equivalent to the \emph{zeroness} problem for $g(n) := \card\Sigma^n - f_L(n)$,
which amounts to decide whether $g = 0$.
This approach has been pioneered by Chomsky and Schützenberger \cite{ChomskySchutzenberger:Algebraic:1963}
who have shown that the generating function
$g_L(x) = \sum_{n = 0}^\infty f_L(n) \cdot x^n$
associated to an unambiguous context-free language $L$ is algebraic (c.f.~\cite{Bousquet-Melou:STACS:2005}).
A similar observation by Stearns and Hunt \cite{StearnsHunt:Unambiguous:SFCS81}
shows that $g_L(x)$ is rational \cite[Chapter 4]{Stanley:1986},
when $L$ is regular,
and more recently by Bostan et al.~\cite{BostanCarayolKoechlinNicaud:ICALP:2020}
who have shown that $g_L(x)$ is holonomic \cite{Stanley:EJC:1980}
when $L$ is recognised by an unambiguous Parikh automaton.
Since the zeroness problem for
rational, algebraic, and holonomic generating functions is decidable,
one obtains decidability of the corresponding universality problems.

\subparagraph{Unambiguous register automata.}

Returning to register automata,
Mottet and Quaas have recently shown that the inclusion problem in the case where $B$ is an unambiguous register automaton over equality data (without guessing) 
can be decided in \TWOEXPSPACE,
and in \EXPSPACE when the numbers of registers of $B$ is fixed \cite[Theorem 1]{MottetQuaas:ToCS:2020}.
%
%
Note that already decidability is interesting,
since unambiguous register automata without guessing
are not closed under complement
in the class of nondeterministic register automata without guessing
\cite[Example 4]{KaminskiZeitlin:IJFCS:2010},
and thus the classical approach via complementing $B$ fails for register automata%
\footnote{
In the more general class of register automata with guessing,
an unproved conjecture proposed by Colcombet states that unambiguous register automata with guessing are effectively closed under complement \cite[Theorem 12]{Colcombet:DCFS:2015},
implying decidability of the universality and containment problems
for unambiguous register automata with guessing and,
a posteriori,
unambiguous register automata without guessing as considered in this paper.
No published proof of this conjecture has appeared as of yet.
}.
(In fact, even for finite automata complementation of unambiguous finite automata
cannot lead to a \PTIME universality algorithm,
thanks to Raskin's recent super-polynomial lower-bound
for the complementation problem for unambiguous finite automata
in the class of non-deterministic finite automata \cite{Raskin:ICALP:2018}).
Mottet and Quaas obtain their result by showing that inclusion can be decided
by checking a reachability property of a suitable graph of triply-exponential size
obtained by taking the product of $A$ and $B$,
and then applying the standard \NLOGSPACE algorithm for reachability in directed graphs.
%

\subparagraph{Our contributions.}

In view of the widespread success of the counting approach to unambiguous models of computation,
one may wonder whether it can be applied to register automata as well.
This is the topic of our paper.
A na\"ive counting approach for register automata
immediately runs into trouble
since there are infinitely many data words 
of length $n$.
The natural remedy is to use the fact that $\A^n$,
albeit infinite,
is \emph{orbit-finite} \cite[Sec.~3.2]{Bojanczyk:AtomBook:2019},
which is a crucial notion generalising finiteness to the realm of relational structures used to model data.
In this way, we naturally count
the number of \emph{orbits} of words/runs of a given length,
which in the context of model theory is sometimes known as the
\emph{Ryll-Nardzewski function} \cite{Schmerl:FM:1978}.
For example, in the case of equality data $(\A, =)$,
the number of orbits of words of length $n$ is the well-known \emph{Bell number $B(n)$},
and for $(\Q, \leq)$ one obtains
the \emph{ordered Bell numbers} (a.k.a.~\emph{Fubini numbers});
c.f.~Cameron's book for more examples \cite[Ch.~7]{Cameron:CUP:2017}.

When considering orbits of runs,
the run length $n$ seems insufficient to obtain recurrence equations.
To this end, we also consider the number of distinct data values $k$
that appear on the word labelling the run.
%
For instance, in the case of equality data,
the corresponding orbit-counting function
is the well-known sequence of \emph{Stirling numbers of the second kind} $\stirling n k : \Q^{\N^2}$,
which satisfies $\stirling 0 0 = 1$, $\stirling m 0 = \stirling 0 m = 0$ for $m \geq 1$, and
\begin{align}
    \label{eq:Stirling}
    \stirling n k = \stirling {n - 1} {k - 1} + k \cdot \stirling {n - 1} k,
    \quad \text{ for } n, k \geq 1.
\end{align}
These intuitions lead us to define the class of
\emph{bidimensional linear recursive sequences with polynomial coefficients} (linrec; c.f.~\eqref{eq:linrec})
which are a class of number sequences in $\Q^{\N^2}$
satisfying a system of shift equations with polynomial coefficients generalising \eqref{eq:Stirling}.
Linrec are sufficiently general to model the orbit-counting functions of register automata
and yet amenable to algorithmic analysis.
Our first result is a complexity upper bound for the zeroness problem for a class of linrec sequences which suffices to model register automata.
\begin{restatable}{theorem}{thmZeronessComplexity}
    \label{thm:zeroness:complexity}
    The zeroness problem for linrec sequences
    with univariate polynomial coefficients from $\Q[k]$
    is in \EXPTIME.
\end{restatable}
\noindent
This is obtained by modelling linrec equations
as systems of linear equations with \emph{skew polynomial coefficients}
(introduced by Ore \cite{Ore:1933})
and then using complexity bounds on the computation of the Hermite normal form
of skew polynomial matrices by Giesbrecht and Kim \cite{GiesbrechtKim:JA:2013}.
Our second result is a reduction of the universality and inclusion problems
to the zeroness problem of a system of linrec equations of exponential size.
Together with \cref{thm:zeroness:complexity},
this yields improved upper bounds on the former problems.
\begin{restatable}{theorem}{thmInclusionComplexity}
    \label{thm:inclusion:complexity}
    The universality $\lang B = (\Sigma \times \A)^*$
    and the inclusion problem $\lang A \subseteq \lang B$
    for register automata $A, B$ without guessing with $B$ unambiguous 
    are in \TWOEXPTIME, and in \EXPTIME for a fixed number of registers of $A, B$.
    The same holds for the equivalence problem $\lang A = \lang B$
    when both automata are unambiguous.
\end{restatable}
The rest of the paper is organised as follows.
In \cref{sec:linrec}, we introduce linrec sequences
(c.f.~\cref{sec:linrec:comparison} for a comparison with well known sequence families from the literature
such as the C-recursive, P-recursive, and the more recent polyrec sequences \cite{CadilhacMazowieckiPapermanPilipczukSenizergues:ICALP:2020}).
In \cref{sec:URA}, we introduce unambiguous register automata
and we present an efficient reduction of the inclusion (and thus equivalence) problem to the universality problem,
which allows us to concentrate on the latter in the rest of the paper.
In \cref{sec:URA2linrec}, we present a reduction of the universality problem to the zeroness problem for linrec.
In \cref{sec:ore}, we show with a simple argument based on elimination
that the zeroness problem for linrec is decidable,
and in \cref{sec:complexity} we derive a complexity upper bound
using non-commutative linear algebra.
Finally, in \cref{sec:conclusions}
we conclude with further work and an intriguing conjecture.
Full proofs, additional definitions, and examples
are provided in \cref{app:linrec,app:URA,app:URA2linrec,app:ore,app:hermite}.


    
    




\subparagraph{Notation.}

Let $\N$, $\Z$, and $\Q$ be the set of non-negative integers, resp., rationals.
The \emph{height} of an integer $k \in \Z$ is $\size k = \abs k$,
and for a rational number $a \in \Q$
uniquely written as $a = \frac p q$ with $p \in \Z, q \in \N$ co-prime
we define $\size a = \max \set {\size p, \size q}$.
Let $\Q[n, k]$ denote the ring of bivariate polynomials.
The \emph{(combined) degree} $\deg P$ of $P = \sum_{i,j} a_{ij}n^ik^j \in \Q[n, k]$
is the maximum $i+j$ s.t.~$a_{ij} \neq 0$
and the \emph{height} $\size P$ is $\max_{i,j}\size {a_{ij}}$.
For a nonempty set $A$ and $n \in \N$,
let $A^n$ be the set of sequences of elements from $A$ of length $n$,
%
In particular, $A^0 = \set\varepsilon$ contains only the empty sequence $\varepsilon$.
Let $A^* = \bigcup_{n \in \N} A^n$ be the set of all finite sequences over $A$.
%
We use the \emph{soft-Oh} notation $\softO {f(n)}$ to denote
$\bigcup_{c \geq 0} O(f(n) \cdot \log^c f(n))$.

\section{Bidimensional linear recursive sequences with polynomial coefficients}
\label{sec:linrec}


Let $f(n, k) : \Q^{\N^2}$ be a bidimensional sequence.
For $L \in \N$, the \emph{first $L$-section} of $f$
is the one-dimensional sequence $f(L, k) : \Q^\N$
obtained by fixing its first component to $L$;
the \emph{second $L$-section} $f(n, L)$ is defined similarly.
The two \emph{shift operators} $\shift 1, \shift 2 : \Q^{\N^2} \to \Q^{\N^2}$ are
\begin{align*}
    (\shift 1 f)(n, k) = f(n+1, k)
        \quad \text{ and } \quad
            (\shift 2 f)(n, k) = f(n, k+1),
                \quad \text{ for all } n, k \geq 0.
\end{align*}
An \emph{affine operator} is a formal expression of the form
$A = p_{00} + p_{01} \cdot \shift 1 + p_{10} \cdot \shift 2$
where $p_{00}, p_{01}, p_{10} \in \Q[n, k]$
are bivariate polynomials over $n, k$ with rational coefficients.
Let $\set{f_1, \dots, f_m}$ be a set of variables denoting bidimensional sequences%
\footnote{We abuse notation
and silently identify variables denoting sequences
with the sequences they denote.}.
A \emph{system of linear shift equations} over $f_1, \dots, f_m$
consists of $m$ equations of the form
\begin{align}
    \label{eq:linrec}
    \left\{
    \begin{array}{rcl}
        \shift 1 \shift 2 f_1 &=& A_{1,1} \cdot f_1 + \cdots + A_{1, m} \cdot f_m, \\
        &\vdots& \\
        \shift 1 \shift 2 f_m &=& A_{m,1} \cdot f_1 + \cdots + A_{m,m} \cdot f_m,
    \end{array}
    \right.
\end{align}
where the $A_{i,j}$'s are affine operators.
A bidimensional sequence $f : \Q^{\N^2}$ is \emph{linear recursive of order $m$, degree $d$, and height $h$} (abbreviated, linrec)
if the following two conditions hold:
\begin{enumerate}[1)]
    \item there are auxiliary bidimensional sequences $f_2, \dots, f_m : \Q^{\N^2}$
    which together with $f = f_1$ satisfy a system of linear shift equations as in \eqref{eq:linrec}
    where the polynomial coefficients have (combined) degree $\leq d$
    and height $\leq h$.
    \item for every $1 \leq i \leq m$
    there are constants denoted $f_i(0, \geq 1), f_i(\geq 1, 0) \in \Q$
    s.t.~$f_i(0, k) = f_i(0, \geq 1)$ and $f_i(n, 0) = f_i(\geq 1, 0)$
    for every $n, k\geq 1$. 
\end{enumerate}
%
%
If we additionally fix 
the initial values $f_1(0, 0), \dots, f_m(0, 0)$,
then the system \eqref{eq:linrec} has a unique solution,
which is computable in \PTIME.
%
%
\begin{lemma}
    \label{lem:linrec:compute}
    The values $f_i(n, k)$'s are computable in deterministic time
    $\softO {m \cdot n \cdot k}$.
\end{lemma}
%
%
%
%


\noindent
In the following we will use the following effective closure under section.

\begin{restatable}{lemma}{lemLinrecSection}
    \label{lem:linrec:section}
    If $f : \Q^{\N^2}$ is linrec of order $\leq m$, degree $\leq d$, and height $\leq h$,
    then its $L$-sections $f(L, k), f(n, L) : \Q^\N$ are linrec
    of order $\leq m \cdot (L+3)$,
    degree $\leq d$,
    and height $\leq h \cdot L^d$.
\end{restatable}

\noindent
We are interested in the following central algorithmic problem for linrec.

\decision{Zeroness problem}
{A system of linrec equations \eqref{eq:linrec} together with all initial conditions.}
{Is it the case that $f_1 = 0$?}
%

\noindent
In \cref{sec:URA2linrec} we use linrec sequences to model the orbit-counting functions of register automata,
which we introduce next.

\section{Unambiguous register automata} 
\label{sec:URA}

We consider register automata over the relational structure $\tuple {\A, =}$
consisting of a countable set $\A$ equipped with equality as the only relational symbol.
Let $\bar a = a_1 \cdots a_n \in \A^n$ be a finite sequence of $n$ data values.
An \emph{$\bar a$-automorphism} of $\A$ is a bijection $\alpha : \A \to \A$
s.t.~$\alpha(a_i) = a_i$ for every $1 \leq i \leq n$,
which is extended pointwise to $\bar a \in \A^n$
and to $L \subseteq \A^*$. 
For $\bar b, \bar c \in \A^n$,
we write $\bar b \sim_{\bar a} \bar c$ whenever there is an $\bar a$-automorphism $\alpha$
\st $\alpha(\bar b) = \bar c$.
The \emph{$\bar a$-orbit} of $\bar b$ is the equivalence class
$\orbit {\bar b}_{\bar a} = \setof{\bar c \in \A^n} {\bar b \sim_{\bar a} \bar c}$,
and the set of $\bar a$-orbits of sequences in $L \subseteq \A^*$ is
$\orbits {\bar a} L = \setof{\orbit {\bar b}_{\bar a}} {\bar b \in L}$.
In the special case when $\bar a = \varepsilon$ is the empty tuple,
we just speak about \emph{automorphism} $\alpha$
and \emph{orbit} $\orbit {\bar b}$.
A set $X$ is \emph{orbit-finite} if $\orbits {} X$ is a finite set \cite[Sec.~3.2]{Bojanczyk:AtomBook:2019}.
%
%
All definitions above extend to $\A_\bot := \A \cup \set \bot$ with $\bot \not\in \A$ in the expected way.
%
%
%
A \emph{constraint} $\varphi$ is a quantifier-free%
\footnote{Since $\tuple{\A, =}$ is a homogeneous relational structure, and thus it admits quantifier elimination,
we would obtain the same expressive power if we would consider more general first-order formulas instead.}
formula generated by 
    $\varphi, \psi ::\equiv
    x = \bot \mid
    x = y \mid
    \varphi \lor \psi \mid
    \varphi \land \psi \mid
    \lnot \varphi$,
%
where $x, y$ are variables
and $\bot$ is a special constant denoting an undefined value.
%
%
%
The semantics of a constraint $\varphi(x_1, \dots, x_n)$ with $n$ free variables $x_1, \dots, x_n$
is the set of tuples of $n$ elements which satisfies:
$\sem\varphi=\setof{a_1, \dots, a_n \in \A_\bot^n}{\A_\bot, x_1:a_1, \dots, x_n:a_n \models \varphi}$.
%
%
A \emph{register automaton} of \emph{dimension} $d \in \N$
is a tuple $\AA = \tuple{d, \Sigma, \L, \L_I, \L_F, \goesto {}}$
where $d$ is the number of registers,
$\Sigma$ is a finite alphabet,
$\L$ is a finite set of \emph{control locations},
of which we distinguish those which are \emph{initial} $\L_I \subseteq \L$, resp.,
\emph{final} $\L_F \subseteq \L$,
and ``$\goesto {}$'' is a set of rules of the form
$p \goesto {\sigma, \varphi} q$,
where $p, q \in \L$ are control locations,
$\sigma \in \Sigma$ is an input symbol from the finite alphabet,
and $\varphi(x_1, \dots, x_d, y, x_1', \dots, x_d')$ is a constraint 
relating the current register values $x_i$'s,
the current input symbol (represented by the variable $y$),
and the next register values of  ${x_i'}$'s.

\begin{example}
    \label{ex:URA}
    Let $\AA$ over $\card\Sigma = 1$ have one register $x$,
    and four control locations $p, q, r, s$,
    of which $p$ is initial and $s$ is final.
    %
    The transitions are
    $p \goesto {x = \bot \land x' = y} q$,
    $p \goesto {x = \bot \land x' = y} r$,
    $q \goesto {x \neq y \land x' = x} q$,
    $q \goesto {x = y \land x' = x} s$,
    $r \goesto {x = y \land x' = x} r$, and
    $r \goesto {x \neq y \land x' = x} s$.
    The automaton accepts all words of the form $a(\A \setminus \set a)^*a$ or $aa^*(\A \setminus \set a)$ with $a \in \A$.
\end{example}

A register automaton is \emph{orbitised} if every constraint $\varphi$
appearing in some transition thereof
denotes an orbit $\sem\varphi \in \orbits {} {\A^{2 \cdot d+1}_\bot}$.
For example, when $d=1$ the constraint $\varphi \equiv x = x'$ is not orbitised,
however $\sem \varphi = \sem {\varphi_0} \cup \sem{\varphi_1}$
splits into two disjoint orbits for the orbitised constraints
$\varphi_0 \equiv x = x' \land x = y$ and
$\varphi_1 \equiv x = x' \land x \neq y$.
The automaton from \cref{ex:URA} is orbitised.
Every register automaton can be transformed in orbitised form 
by replacing every transition $p \goesto {\sigma, \varphi} q$
with exponentially many transitions $p \goesto {\sigma, \varphi_1} q, \dots, p \goesto {\sigma, \varphi_n} q$,
for each orbit $\sem{\varphi_i}$ of $\sem \varphi \subseteq \A^{2 \cdot d+1}_\bot$.
%

A \emph{register valuation} is a tuple of (possibly undefined) values
$\bar a = \tuple {a_1, \dots, a_d} \in \A^d_\bot$.
A \emph{configuration} is a pair $\tuple {p, \bar a}$,
where $p \in \L$ is a control location and $\bar a \in \A^d_\bot$ is a register valuation;
it is \emph{initial} if $p \in \L_I$ is initial and all registers are initially undefined $\bar a = \tuple{\bot, \dots, \bot}$,
and it is \emph{final} whenever $p \in \L_F$ is so.
The \emph{semantics} of a register automaton $\AA$
is the infinite transition system $\sem\AA = \tuple{C, C_I, C_F, \goesto {}}$
where $C$ is the set of configurations, of which $C_I, C_F \subseteq C$ are the initial, resp., final ones,
and ${\goesto {}} \subseteq C \times (\Sigma \times \A) \times C$
is the set of all transitions of the form
\begin{align*}
    \tuple {p, \bar a} \goesto {\sigma, a} \tuple {q, \bar a'},
    \qquad \text{with } \sigma \in \Sigma, a \in \A, \text{ and } \bar a, \bar a' \in \A^d_\bot,
\end{align*}
\st there exists a rule $p \goesto{\sigma, \varphi} q$
where satisfying the constraint $\A_\bot, \bar x : \bar a, y : a, \bar x' : \bar a' \models \varphi$.
%
%
A \emph{data word} is a sequence $w = \tuple {\sigma_1, a_1} \cdots \tuple {\sigma_n, a_n} \in (\Sigma \times \A)^*$.
A \emph{run over} a data word $w$
\emph{starting at} $c_0 \in C$
and \emph{ending at} $c_n \in C$
is a sequence $\pi$ of transitions of $\sem \AA$ of the form
%
 $\pi = c_0 \goesto {\sigma_1, a_1} c_1 \goesto {\sigma_2, a_2} \cdots \goesto {\sigma_n, a_n} c_n.$
%
We denote with $\runs {c_0} {w} {c_n}$ the set of runs over $w$
starting at $c_0$ and ending in $c_n$,
and with $\runs {C_I} {w} {c_n}$ the set of \emph{initial runs},
i.e., those runs over $w$ starting at some initial configuration $c_0 \in C_I$
and ending in $c_n$.
The run $\pi$ is \emph{accepting} if $c_n \in C_F$.
The language $\lang {\AA, c}$ recognised from configuration $c \in C$
is the set of data words labelling some accepting run starting at $c$;
the language recognised from a set of configurations $D \subseteq C$ is
$\lang {\AA, D} = \bigcup_{c \in D} \lang {\AA, c}$,
and the language recognised by the register automaton $\AA$ is $\lang \AA = \lang {\AA, C_I}$.
Similarly, the \emph{backward language} $\blang {\AA, c}$
is the set of words labelling some run starting at an initial configuration and ending at $c$.
Thus, we also have $\lang \AA = \blang {\AA, C_F}$.
%
%
A register automaton is \emph{deterministic} if for every input word
there exists at most one initial run,
and \emph{unambiguous} if for every input word
there is at most one initial and accepting run.
A register automaton is \emph{without guessing}
if, for every initial run $\tuple {p, \bot^d} \goesto w \tuple {q, \bar a}$
every non-$\bot$ data value in $\bar a$ occurs in the input $w$, written $\bar a \subseteq w$.
In the rest of the paper we will study exclusively automata without guessing.
A deterministic automaton is unambiguous and without guessing.
These semantic properties can be decided in \PSPACE with simple reachability analyses (c.f.~\cite{Colcombet:DCFS:2015}).

\begin{example}
    The automaton from \cref{ex:URA} is unambiguous and without guessing.
    An example of language which can only be recognised by ambiguous register automata
    is the set of words where the same data value appears two times
    $L = \setof{u \cdot a \cdot v \cdot a \cdot w} {a \in \A; u,v,w \in \A^*}$.
\end{example}

\begin{lemma}
    \label{lem:URA:bijection}
    If $\AA$ is an unambiguous register automaton,
    then there is a bijection between the language it recognises
    $\lang \AA = \lang {\AA, C_I} = \blang {\AA, C_F}$
    and the set of runs starting at some initial configuration in $C_I$
    and ending at some final configuration in $C_F$.
\end{lemma}
%
%
We are interested in the following decision problem.



\decision{Inclusion problem}
{Two register automata $\AA, \BB$ over the same input alphabet $\Sigma$.}
{Is it the case that $\lang \AA \subseteq \lang \BB$?}

\noindent
The \emph{universality problem} asks $\lang \AA = (\Sigma \times \A)^*$,
and the \emph{equivalence problem} $\lang \AA = \lang \BB$.
%
%
In general, universality reduces to equivalence, which in turn reduces to inclusion.
In our context, inclusion reduces to universality
and thus all three problems are equivalent.

\begin{restatable}{lemma}{lemReductionToUniversality}
    \label{lem:reduction:to:universality}
    Let $\AA$ and $\BB$ be two register automata.
    \begin{enumerate}
        \item The inclusion problem $\lang\AA \subseteq \lang\BB$ with $\AA$ orbitised and without guessing
        reduces in \PTIME to the case where $\AA$ is deterministic.
        The reduction preserves whether $\BB$ is 1) unambiguous,
        2) without guessing, and 3) orbitised.
        
        \item The inclusion problem $\lang\AA \subseteq \lang\BB$ with $\AA$ deterministic
        reduces in \PTIME to the universality problem for some register automaton $\CC$.
        If $\BB$ is unambiguous, then so is $\CC$.
        If $\BB$ is without guessing, then so is $\CC$.
        If $\AA$ and $\BB$ are orbitised, then so is $\CC$.
    \end{enumerate}
\end{restatable}


\section{Universality of unambiguous register automata without guessing}
\label{sec:URA2linrec}

We reduce universality of unambiguous register automata without guessing
to zeroness of bidimensional linrec sequences with univariate polynomial coefficients.
The \emph{width} of a sequence of data values $\bar a = a_1 \cdots a_n \in \A^n$
is $\width {\bar a} = \card{\set{a_1, \dots, a_n}}$,
for a word $w = \tuple{\sigma_1, a_1} \cdots \tuple{\sigma_n, a_n} \in (\Sigma \times \A)^*$
we set $\width w = \width {(a_1 \cdots a_n)}$,
and for a run $\pi$ over $w$ we set $\width \pi = \width w$.
Let the \emph{Ryll-Nardzewski function} $\RNR {p, \bar a} (n, k)$
of a configuration $\tuple{p, \bar a} \in C = \L \times \A^d_\bot$
count the number of $\bar a$-orbits of initial runs of length $n$ and width $k$
ending in $\tuple{p, \bar a}$:
\begin{align}
    \label{eq:Ryll-Nardzewski}
    \RNR {p, \bar a} (n, k) =
        \card {\setof{\orbit {\pi}_{\bar a}}{w \in (\Sigma \times \A)^n, \pi \in \runs {C_I} w {p, \bar a}, \width w = k}}.
\end{align}
\begin{restatable}{lemma}{lemRNRorbit}
    \label{lem:RNR:orbit}
    Let $\bar a, \bar b \in \A_\bot^d$.
    If $\orbit {\bar a} = \orbit {\bar b}$,
    then $\RNR {p, \bar a} (n, k) = \RNR {p, \bar b} (n, k)$
    for every $n, k \geq 0$.
\end{restatable}
\noindent
We thus overload the notation and write $\RNR {p, \orbit {\bar a}}$
instead of $\RNR {p, \bar a}$.
Since $\A^d_\bot$ is orbit-finite, this yields finitely many variables $\RNR {p, \orbit {\bar a}}$'s.
By slightly abusing notation, let
$\RNR {C_F} (n, k) = \sum_{\orbit{\tuple{p, \bar a}} \in \orbits {} {C_F}} \RNR {p, \orbit{\bar a}} (n, k)$
be the sum of the Ryll-Nardzewski function over all orbits of accepting configurations.
When the automaton is unambiguous,
thanks to 
\cref{lem:URA:bijection},
$\RNR {C_F} (n, k)$ is also the number of orbits of accepted words of length $n$ and width $k$.

\begin{lemma}
    \label{lem:univ-iff-count-URA}
    Let $\AA$ be an unambiguous register automaton w/o guessing over $\Sigma$
    and let $S_\Sigma(n, k)$ be the number of orbits of all words of length $n$ and width $k$.
    We have
        $\lang \AA = (\A \times A)^*$
            if, and only if, 
                $\forall n, k \in \N \cdot \RNR {C_F} (n, k) = S_\Sigma (n, k)$.
\end{lemma}
In other words, universality of $\AA$
reduces to zeroness of
{${\RNR {} := S_{\Sigma} - \RNR {C_F}}$}.
The sequence $S_{\Sigma}$ is linrec 
since it satisfies the recurrence in \cref{fig:linrec:equations}
with initial conditions
$S_\Sigma(0, 0) = 1$ and $S_\Sigma(n+1, 0) = S_\Sigma(0, k+1)= 0$
for $n, k \geq 0$.
%
We show that all the sequences of the form $\RNR {p, \orbit {\bar a}}$ are also linrec
and thus also $\RNR {}$ will be linrec.
%
%
\begin{figure}
    \centering
    \includegraphics[scale=1]{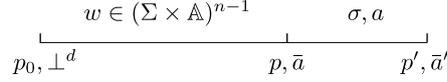}
    \caption{Last-step decomposition.}
    \label{fig:last-step-decompositiion}
\end{figure}
%
%
We perform a last-step decomposition of an initial run;
c.f.~\cref{fig:last-step-decompositiion}.
Starting from some initial configuration $\tuple{p_0, \bot^d}$,
the automaton has read a word $w$ of length $n-1$
leading to $\tuple{p, \bar a}$.
Then, the automaton reads the last letter $\tuple{\sigma, a}$
and goes to $\tuple{p', \bar a'}$
via the transition
$t = (p, \bar a \goesto {\sigma, a} p', \bar a')$.
The question is in how many distinct ways can an orbit of the run over $w$
be extended into an orbit of the run over $w\cdot \tuple{\sigma, a}$.
We distinguish three cases.
\begin{enumerate}
    \item[\I:] Assume that $a$ appears in register $\bar a_i = a$.
    Since the automaton is without guessing, $a \in w$ has appeared earlier in the input word and $\bar a' \subseteq \bar a$
    (ignoring $\bot$'s).
    Thus, each $\bar a$-orbit of runs
    $\orbit {p_0, \bot^d \goesto w p, \bar a}_{\bar a}$
    yields, via the fixed $t$,
    an $\bar a'$-orbit of runs
    $\orbit {p_0, \bot^d \goesto w p, \bar a \goesto {\sigma, a} p', \bar a'}_{\bar a'}$ of the same width
    in just one way.
    \item[\II:] Assume that $a$ is globally fresh $a \not\in w$,
    and thus in particular $a \not\in \bar a$
    since the automaton is without guessing.
    Each $\bar a$-orbit of runs $\orbit {p_0, \bot^d \goesto w p, \bar a}_{\bar a}$ of width $\width w$
    yields, via the fixed $t$, a single $\bar a'$-orbit of runs
    $\orbit {p_0, \bot^d \goesto w p, \bar a \goesto {\sigma, a} p', \bar a'}_{\bar a'}$
    of width $\width {(w \cdot a)} = \width w + 1$.
    \item[\III:] Assume that $a \in w$ is not globally fresh,
    but it does not appear in any register $a \not\in \bar a$.
    Since the automaton is without guessing,
    every value in $\bar a$ appears in $w$.
    Consequently, $a$ can be any of the $\width w$ distinct values in $w$,
    with the exception of $\width {\bar a}$ values.
    Each $\bar a$-orbit of runs
    $\orbit {p_0, \bot \goesto w p, \bar a}_{\bar a}$ of width $\width w$
    yields $\width w - \width {\bar a} \geq 0$ $\bar a'$-orbits of runs
    $\orbit {p_0, \bot^d \goesto w p, \bar a \goesto {\sigma, a} p', \bar a'}_{\bar a'}$ of the same width.
\end{enumerate}
(As expected, we do not need unambiguity at this point,
since we are counting orbits of runs.)
%
We obtain the equations in \cref{fig:linrec:equations},
where the sums range over orbits of transitions.
This set of equations is finite since there are finitely many orbits
$\orbit {\bar a} \in \orbits {} {\A^d_\bot}$ of register valuations,
and moreover we can effectively represent each orbit by a constraint \cite[Ch.~4]{Bojanczyk:AtomBook:2019}.
Strictly speaking, the equations are not linrec
due to the ``$\max$'' operator,
however they can easily be transformed to linrec
by considering $G_{p, \orbit{\bar a}}(n, K)$ separately for $1 \leq K < d$;
in the interest of clarity,
we omit the full linrec expansion.
The initial condition is
$\RNR {p, \orbit{\bar a}} (0, 0) = 1$ if $p \in I$ initial,
and $\RNR {p, \orbit{\bar a}} (0, 0) = 0$ otherwise.
The two $0$-sections satisfy 
$\RNR {p, \orbit{\bar a}} (n+1, 0) = 0$ for $n \geq 0$
(if the word is nonempty, then there is at least one data value)
and $\RNR {p, \orbit{\bar a}} (0, k+1) = 0$ for $k \geq 0$
(an empty word does not have any data value).
\begin{figure}
    \begin{align*}
        \RNR {p', \orbit {\bar a '}} (n+1, k+1) =
            &\ \sum_{\orbit {p, \bar a \goesto {\sigma, a} p', \bar a'}:\; a \in \bar a}
                    \underbrace{\RNR {p, \orbit{\bar a}} (n, k+1)}_{\I}\ + \\
            &\ \sum_{\orbit {p, \bar a \goesto {\sigma, a} p', \bar a'}:\; a \not\in \bar a}
                \left( \underbrace{\RNR {p, \orbit {\bar a}} (n, k)}_{\II} + \underbrace{\max (k + 1 - \width {\orbit {\bar a}}, 0) \cdot \RNR {p, \orbit {\bar a}} (n, k+1)}_{\III}\right), \\
        S_\Sigma(n+1, k+1) =
            &\ {\card \Sigma} \cdot S_{\Sigma}(n, k) + {\card \Sigma} \cdot (k + 1) \cdot S_{\Sigma}(n, k+1), \\
    \RNR {} (n, k) =
        &\ S_{\Sigma}(n, k) - \sum_{\orbit{p, \bar a} \in \orbits {} {C_F}} \RNR {p, \orbit{\bar a}} (n, k).
    \end{align*}
    \caption{Linrec automata equations.}
    \label{fig:linrec:equations}
\end{figure}
%
%
\begin{restatable}{lemma}{lemURALinrec}
    \label{lem:URA:linrec}
    The sequences $\RNR {p, \orbit {\bar a}}$'s
    satisfy the system of equations in \cref{fig:linrec:equations}.
\end{restatable}
\begin{example}
    \label{ex:equations}
    The equations corresponding to the automaton in \cref{ex:URA}
    are as follows.
    (Since the automaton is orbitised, we can omit the orbit.)
    We have $G_p(0, 0) = 1$, $G_q(0, 0) = G_r(0, 0) = G_s(0, 0) = 0$
    and for $n, k \geq 0$:
    %
    \begin{align*}
        G_p(n+1, k+1) &= 0, \\
        G_q(n+1, k+1) &= \underbrace{G_p(n, k)}_{\II} + \underbrace{(k+1) \cdot G_p(n, k+1)}_{\III} + \underbrace{G_q(n, k)}_{\II} + \underbrace{k \cdot G_q(n, k+1)}_{\III}, \\
        G_r(n+1, k+1) &= \underbrace{G_p(n, k)}_{\II} + \underbrace{(k+1) \cdot G_p(n, k+1)}_{\III} + \underbrace{G_r(n, k+1)}_{\I}, \\
        G_s(n+1, k+1) &= \underbrace{G_q(n, k+1)}_{\I} + \underbrace{G_r(n, k)}_{\II} + \underbrace{k \cdot G_r(n, k+1)}_{\III}.
    \end{align*}
\end{example}
\begin{restatable}{lemma}{lemUniversalityToZeroness}
    \label{lem:universality2zeroness}
    Let $\AA$ be an unambiguous register automaton over equality atoms without guessing
    with $d$ registers and $\ell$ control locations.
    The universality problem for $\AA$
    reduces to the zeroness problem of the linrec sequence $G$
    defined by the system of equations in \cref{fig:linrec:equations}
    containing $O(\ell \cdot 2^{d \cdot \log d})$
    variables and equations
    and constructible in $\PSPACE$.
    If $\AA$ is already orbitised,
    then the system of equations has size $O(\ell)$.
\end{restatable}


\section{Decidability of the zeroness problem}
\label{sec:ore}

In this section,
we present an algorithm to solve the zeroness problem of bidimensional linrec sequences
with univariate polynomial coefficients,
which is sufficient for linrec sequences from \cref{fig:linrec:equations}.
We first give a general presentation on elimination for bivariate polynomial coefficients,
and then we use the univariate assumption to obtain a decision procedure.
%
%
We model the non-commutative operators appearing in the definition of linrec sequences \eqref{eq:linrec}
with Ore polynomials (a.k.a.~skew polynomials) \cite{Ore:1933}%
\footnote{The general definition of the Ore polynomial ring $R[\partial; \sigma, \delta]$
uses an additional component $\delta : R \to R$ in order to model differential operators.
We present a simplified version which is enough for our purposes.}. 
Let $R$ be a (not necessarily commutative) ring
and $\sigma$ an automorphism of $R$.
The ring of \emph{(shift) skew polynomials} $R[\partial; \sigma]$
is defined as the ring of polynomials
but where the multiplication operation satisfies the following commutation rule:
For a coefficient $a\in R$ and the unknown $\partial$, we have
$$\partial \cdot a = \sigma(a) \cdot \partial.$$
(The usual ring of polynomials is recovered when $\sigma$ is the identity.)
The multiplication extends to monomials as
$a\partial^k \cdot b \partial^l = a \sigma^k(b) \cdot \partial^{k+l}$
and to the whole ring by distributivity.
The \emph{degree} of a skew monomial $a \cdot \partial^k$ is $k$,
and the degree $\deg P$ of a skew polynomial $P$
is the maximum of the degrees of its monomials.
The degree function satisfies the expected identities
$\deg (P \cdot Q) = \deg P + \deg Q$ and
$\deg (P + Q) \leq \max (\deg P, \deg Q)$.
A skew polynomial is \emph{monic} 
if the coefficient of its monomial of highest degree is $1$.
The crucial and only property that we need in this section
is that skew polynomial rings admit a Euclidean pseudo-division algorithm,
which in turns allows one to find common left multiples.
A skew polynomial ring $R[\partial; \sigma]$ has \emph{pseudo-division}
if for any two skew polynomials $A, B \in R[\partial; \sigma]$ with $\deg A \geq \deg B$
there is a coefficient $a \in R$ and skew polynomials $Q, R \in R[\partial; \sigma]$
s.t.~$a \cdot A = P \cdot B + Q$ and $\deg Q < \deg B$.
%
We say that a ring $R$ has the \emph{common left multiple} (\CLM) property
if for every $a,b\neq 0$, there exists $c,d\neq 0$ such that
$c \cdot a = d \cdot b$.
\begin{restatable}[\protect{c.f.~\cite[Sec.~1]{Ore:1931}}]{theorem}{thmCLM}
    \label{thm:CLM}
    If $R$ has the \CLM property, then
    \begin{inparaenum}[1)]
        \item $R[\partial;\sigma]$ has a pseudo-division, and
        \item $R[\partial;\sigma]$ also has the \CLM property.
    \end{inparaenum}
\end{restatable}
%
%
\noindent
The most important instances of skew polynomials are
the \emph{first} and \emph{second Weyl algebras}:
\begin{align}
    \label{eq:Weyl:algebras}
    W_1 = \Q[n, k][\shift 1; \sigma_1]
        \quad \text{and} \quad
            W_2 = W_1[\shift 2; \sigma_2] = \Q[n, k][\shift 1; \sigma_1][\shift 2; \sigma_2],
\end{align}
where $\Q[n, k]$ is the ring of bivariate polynomials,
and the shifts satisfy $\sigma_1 (p(n, k)) := p(n+1, k)$
%
%
%
and $\sigma_2\left(\sum_i p_i(n, k)\shift 1^i\right) := \sum_i p_i(n, k+1) \shift 1^i$.
Skew polynomials in $W_2$ act on bidimensional sequences $f:\Q^{\N^2}$
by interpreting $\shift1$ and $\shift2$ as the two shifts.
A linrec system of equations \eqref{eq:linrec}
can thus be interpreted as a system of linear equations with variables $f_1, \dots, f_m$
and coefficients in $W_2$.




\newcommand\equationsExWeyl{
\begin{align*}
    &\begin{array}{rrrrrrrr}
        \shift 1\shift 2 \cdot G_p&&&& = 0, \\
        -(1 + (k+1) \shift 2) \cdot G_p & + (\shift 1\shift 2 - k \shift 2 - 1) \cdot G_q&&& = 0, \\
        -(1 + (k+1) \shift 2) \cdot G_p & & + (\shift 1 \shift 2 - \shift 2) \cdot G_r&& = 0, \\
        & - \shift 2 \cdot G_q & - (1 + k \shift 2) \cdot G_r& + \shift 1 \shift 2 \cdot G_s& = 0,
    \end{array} \\
    &(\shift 1 \shift 2 - (k+1)\shift 2 - 1) \cdot S_1 = 0, \\
    &G_s - S_1 + G = 0.
\end{align*}
}

\begin{example}
    \label{ex:weyl:equations}
    Continuing our running \cref{ex:equations},
    we obtain the following linear system of equations with $W_2$ coefficients:
    %
        %
    \equationsExWeyl
\end{example}
%
%
%
%
%
\noindent
Since $W_0 = \N[n, k]$ is commutative, 
it obviously has the \CLM property.
By two applications of \cref{thm:CLM}, we have
(see \cref{app:CLM:examples} for \CLM examples):

\begin{corollary}
    \label{cor:CLM}
    The two Weyl algebras $W_1$ and $W_2$ have the \CLM property.
\end{corollary}


%
A (linear) \emph{cancelling relation} (\CR) for a bidimensional sequence $f : \Q^{\N^2}$
is a linear equation of the form
\begin{align}
    \label{eq:cancelling:relation}
    \tag{\CR-2}
    p_{i^*, j^*}(n, k)\cdot \shift1^{i^*} \shift2^{j^*} f =
        \sum_{(i, j) <_\text{lex} (i^* , j^*)} p_{i, j}(n, k)\cdot \shift1^i \shift2^j f,
\end{align}
where $p_{i^*, j^*}(n, k), p_{i, j}(n, k)\in\Q[n, k]$ are bivariate polynomial coefficients
and $<_\text{lex}$ is the lexicographic ordering.
%
%
Cancelling relations for a one-dimensional sequence $g : \Q^\N$
are defined analogously (we use the second variable $k$ as the index for convenience):
\begin{align}
    \label{eq:cancelling:relation:1}
    \tag{\CR-1}
    q_{j^*}(k) \cdot \shift2^{j^*} g =
        \sum_{0 \leq j < j^*} q_j(k)\cdot \shift2^j g.
\end{align}
We use cancelling relations as certificates of zeroness for $f$
when the $p_{i, j}$'s are univariate.
We do not need to construct any cancelling relation, just knowing that some exists with the required bounds suffices.

\begin{restatable}{lemma}{lemZeroness}
    \label{lem:zeroness}
    The zeroness problem for a bidimensional linrec sequence $f : \Q^{\N^2}$ of order $\leq m$
    and univariate polynomial coefficients in $\Q[k]$
    admitting some cancelling relation \eqref{eq:cancelling:relation}
    with leading coefficient $p_{i^*, j^*}(k)\in \Q[k]$
    of degree $\leq e$ and height $\leq h$
    s.t.~each of the one-dimensional sections $f(M, k) \in \Q^\N$ for $1 \leq M \leq i^*$
    also admits some cancelling relation \eqref{eq:cancelling:relation:1}
    of $\shift2$-degree $\leq d$
    with leading polynomial coefficients
    of degrees $\leq e$ and height $\leq h$
    is decidable in deterministic time $\softO {p(m, i^*, j^*, d, e, h)}$ for some polynomial $p$.
\end{restatable}

%
Elimination already yields decidability with elementary complexity
for the zeroness problem and thus for the universality/equivalence/inclusion problems of unambiguous register automata without guessing.

\begin{restatable}{theorem}{thmZeronessLinrec}
    \label{thm:zeroness:linrec}
    The zeroness problem for linrec sequences with univariate polynomial coefficients from $\Q[k]$ (or from $\Q[n]$) is decidable.
\end{restatable}
%
%

\begin{example}
    Continuing our running \cref{ex:weyl:equations},
    we subsequently eliminate $G_p, G_s, G_r, G_q, S$
    finally obtaining (c.f.~\cref{ex:elimination} in \cref{app:CR:examples} for details)
    \begin{align}
    \label{eq:cancelling:relation:example}
        \begin{array}{ll}
            G(n+4, k+4) =
                & (k+3) \cdot G(n+3, k+4) + G(n+3, k+3)\; + \\
                &- (k+2) \cdot G(n+2, k+4) - G(n+2, k+3).
        \end{array}
    \end{align}
    As expected, all coefficients are polynomials in $\Q[k]$
    and in particular they do not involve the variable $n$.
    Moreover, we note that the relation above is \emph{monic},
    in the sense that the lexicographically leading term $G(n+4, k+4)$ has coefficient $1$ (c.f.~\cref{sec:conclusions}).
    %
    %
    %
    (C.f.~\cref{ex:CR:two-registers} for elimination in a two-register automaton
    and \cref{ex:URA:universal} for a one-register automaton accepting all words of length $\geq 2$.)
\end{example}
%
%
\noindent
We omit a precise complexity analysis of elimination
because better bounds can be obtained
by resorting to linear non-commutative algebra,
which is the topic of the next section.

\section{Complexity of the zeroness problem}
\label{sec:complexity}

In this section we present an \EXPTIME algorithm to solve the zeroness problem
and we apply this result to register automata.
We compute the \emph{Hermite normal form} (\HNF)
of the matrix with skew polynomial coefficients associated to \eqref{eq:linrec}
in order to do elimination in a more efficient way.
The complexity bounds provided by Giesbrecht and Kim \cite{GiesbrechtKim:JA:2013}
on the computation of the \HNF
lead to the following bounds for cancelling relations;
c.f.~\cref{app:hermite} for further details and full proofs.

\begin{restatable}{lemma}{lemCRBounds}
    \label{lem:CR:bounds}
    A linrec sequence $f \in \Q^{\N^2}$ of order $\leq m$,
    degree $\leq d$,
    and height $\leq h$
    admits a cancelling relation \eqref{eq:cancelling:relation}
    with the orders $i^*, j^*$ and the degree of $p_{i^*, j^*}$ polynomially bounded,
    and with height $\size {p_{i^*, j^*}}$ exponentially bounded.
    Similarly, its one-dimensional sections $f(0, k), \dots, f(i^*, k) \in \Q^\N$
    also admit cancelling relations \eqref{eq:cancelling:relation:1}
    of polynomially bounded orders and degree,
    and exponentially bounded height.
\end{restatable}
\noindent
This allows us to prove below the \EXPTIME upper-bound for zeroness of \cref{thm:zeroness:complexity},
and the \TWOEXPTIME algorithm for inclusion of \cref{thm:inclusion:complexity}.
\begin{proof}[Proof of \cref{thm:zeroness:complexity}]
    Thanks to the bounds from \cref{lem:CR:bounds},
    $i^*, j^*$ are polynomially bounded;
    we can find a polynomial bound $d$ on the $\shift2$-degrees
    of the cancelling relations $R_0, \dots, R_{i^*}$ for the sections $f(0, k), \dots, f(i^*, k)$, respectively;
    we can find a polynomial bound $e$ on the degrees of $p_{i^*, j^*}(k)$
    and the leading polynomial coefficients of the $R_i$'s;
    and an exponential bound $h$ on $\size {p_{i^*, j^*}}$
    and the heights of the leading polynomial coefficients of the $R_i$'s.
    %
    We thus obtain an \EXPTIME algorithm by \cref{lem:zeroness}.
\end{proof}
\noindent
This yields the announced upper-bounds for the inclusion problem for register automata.

\begin{proof}[Proof of \cref{thm:inclusion:complexity}]
    For the universality problem $\lang B = (\Sigma \times \A)^*$,
    let $d$ be the number of registers and $\ell$ the number of control locations of $B$.
    By \cref{lem:universality2zeroness},
    the universality problem reduces in \PSPACE
    to zeroness of a linrec system with
    polynomial coefficients in $\Q[k]$
    containing $O(\ell \cdot 2^{d \cdot \log d})$ variables $G_{p, \orbit {\bar a}}$
    and the same number of equations.
    By \cref{thm:zeroness:complexity}, we get a \TWOEXPTIME algorithm.
    When the numbers of registers $d$ is fixed, we get an \EXPTIME algorithm.
    For the inclusion problem $\lang A \subseteq \lang B$,
    we first orbitise $A$ into an equivalent orbitised register automaton without guessing $A'$.
    %
    %
    A close inspection of the two constructions leading to $C$ in the proof of \cref{lem:reduction:to:universality}
    reveal that transitions in $C$ are either transitions from $A'$
    (and thus already orbitised),
    or pairs of a transition in $B$ together with a transition in $A'$,
    the second of which is already orbitised.
    It follows that orbitising $C$ incurs in an exponential blow-up w.r.t.~the number of registers of $B$,
    but only polynomial w.r.t.~the number of registers of $A'$ (and thus of $A$),
    since the $A'$-part in $C$ is already orbitised.
    Consequently, we can write (in \PSPACE)
    a system of linrec equations for the universality problem of $C$
    of size exponential in the number of registers of $A$ and of $B$.
    By reasoning as in the first part of the proof,
    we obtain a \EXPTIME algorithm for the universality problem of $C$,
    and thus a \TWOEXPTIME algorithm for the original inclusion problem $\lang A \subseteq \lang B$.
    If both the number of registers of $A$ and of $B$ is fixed,
    we get an \EXPTIME algorithm.
    The equivalence problem $\lang A = \lang B$ with both automata $A, B$ unambiguous
    reduces to two inclusion problems.
\end{proof}

\section{Further remarks and conclusions}
\label{sec:conclusions}

We say that $P = \sum_{i, j} p_{i, j}(n, k) \cdot \shift1^i \shift2^j$
is \emph{monic} if $p_{i^*, j^*} = 1$
where $(i^*, j^*)$ is the lexicographically largest pair $(i, j)$ s.t.~$p_{i, j} \neq 0$.
The cancelling relation \eqref{eq:cancelling:relation} in our examples \eqref{eq:cancelling:relation:example},
\eqref{eq:CR:two-registers},
\eqref{eq:CR:universal},
\eqref{eq:cancelling:relation:example:2}
happens to be monic in this sense.
\begin{conjecture}[Monicity conjecture]
    \label{conjecture:monic}
    There always exists a \emph{monic} cancelling relation \eqref{eq:cancelling:relation}
    for linrec systems obtained from automata equations in \cref{fig:linrec:equations},
    and similarly for their sections \eqref{eq:cancelling:relation:1}.
\end{conjecture}
\cref{conjecture:monic} has important algorithmic consequences.
The exponential complexity in \cref{thm:zeroness:complexity}
comes from the exponential growth of the rational number coefficients (heights) in the \HNF.
This is due to the use of \cref{lem:zeroness},
whose complexity depends on the maximal root of the leading polynomial $p_{i^*, j^*}(n, k)$ from \eqref{eq:cancelling:relation}.
If \cref{conjecture:monic} holds,
then $p_{i^*, j^*}(n, k) = 1$,
\cref{lem:zeroness} would yield a \PTIME algorithm for zeroness,
and consequently all complexities in \cref{thm:inclusion:complexity},
would drop by one exponential.
This provides ample motivation to investigate the monicity conjecture.

In order to obtain the lower \EXPTIME complexity for 
$\lang A \subseteq \lang B$ in \cref{thm:inclusion:complexity}
we have to fix the number of registers in \emph{both} automata $A$ and $B$.
The \EXPSPACE upper bound of Mottet and Quaas \cite{MottetQuaas:ToCS:2020}
holds already when only the number of registers of $B$ is fixed,
while we only obtain a \TWOEXPTIME upper bound in this case.
It is left for future work whether the counting approach can yield better bounds without fixing the number of registers of $A$.


The fact that the automata are non-guessing is crucial in each of the cases \I, \II, and \III~of the equations in \cref{fig:linrec:equations} in order to correctly count the number of orbits of runs.
For automata with guessing
from the fact that the current input $a$ is stored in a register we cannot deduce that $a$ actually appeared previously in the input word $w$,
and thus our current parametrisation in terms of length and width does not lead to a recursive characterisation.

in the last-step decomposition since we need to know that all values in $\bar a$

Finally, it is also left for further work
to extend the counting approach to other data domains such as total order atoms, random graph atoms, etc\dots,
and, more generally, to arbitrary homogeneous and $\omega$-categorical atoms under suitable computability assumptions (c.f.~\cite{ClementeLasota:CSL:2015}),
and to other models of computation such as register pushdown automata \cite{ChengKaminski:CFL:AI98,MurawskiTzevelekos:JCSS:2017}.

\newpage
\bibliographystyle{plainurl} 
\bibliography{bibliography}

\begin{thebibliography}{10}

\bibitem{AlterKubota:JCT:1973}
Ronald Alter and K.K Kubota.
\newblock Prime and prime power divisibility of {Catalan} numbers.
\newblock {\em Journal of Combinatorial Theory, Series A}, 15(3):243 -- 256,
  1973.

\bibitem{BaierKieferKleinKluppelholzMullerWorrell:CAV:2016}
Christel Baier, Stefan Kiefer, Joachim Klein, Sascha Kl{\"u}ppelholz, David
  M{\"u}ller, and James Worrell.
\newblock {Markov Chains and Unambiguous B{\"u}chi Automata}.
\newblock In Swarat Chaudhuri and Azadeh Farzan, editors, {\em Proc. of
  CAV'16}, pages 23--42, Cham, 2016. Springer International Publishing.

\bibitem{BenediktDuffSharadWorrell:PolyAut:LICS:2017}
M.~Benedikt, T.~Duff, A.~Sharad, and J.~Worrell.
\newblock Polynomial automata: Zeroness and applications.
\newblock In {\em Proc. of LICS'17}, pages 1--12, June 2017.
\newblock \href {https://doi.org/10.1109/LICS.2017.8005101}
  {\path{doi:10.1109/LICS.2017.8005101}}.

\bibitem{Bojanczyk:AtomBook:2019}
Miko{\l}aj Boja{\'n}czyk.
\newblock {\em Slightly Infinite Sets}.
\newblock 2019.
\newblock URL: \url{https://www.mimuw.edu.pl/~bojan/paper/atom-book}.

\bibitem{BostanCarayolKoechlinNicaud:ICALP:2020}
Alin Bostan, Arnaud Carayol, Florent Koechlin, and Cyril Nicaud.
\newblock {Weakly-Unambiguous Parikh Automata and Their Link to Holonomic
  Series}.
\newblock In Artur Czumaj, Anuj Dawar, and Emanuela Merelli, editors, {\em
  Proc. of ICALP'20}, volume 168 of {\em LIPIcs}, pages 114:1--114:16,
  Dagstuhl, Germany, 2020. Schloss Dagstuhl--Leibniz-Zentrum f{\"u}r
  Informatik.

\bibitem{BostanChyzakSalvyLi:ISAAC:2012}
Alin Bostan, Fr\'{e}d\'{e}ric Chyzak, Bruno Salvy, and Ziming Li.
\newblock Fast computation of common left multiples of linear ordinary
  differential operators.
\newblock In {\em Proc. of ISAAC'12}, pages 99--106, New York, NY, USA, 2012.
  ACM.

\bibitem{BousquetLoeding:Unambiguous:LATA:2010}
Nicolas Bousquet and Christof L{\"o}ding.
\newblock Equivalence and inclusion problem for strongly unambiguous b{\"u}chi
  automata.
\newblock In Adrian-Horia Dediu, Henning Fernau, and Carlos Mart{\'\i}n-Vide,
  editors, {\em Proc. of LATA'10}, pages 118--129, Berlin, Heidelberg, 2010.
  Springer Berlin Heidelberg.

\bibitem{Bousquet-Melou:STACS:2005}
Mireille Bousquet-M{\'e}lou.
\newblock Algebraic generating functions in enumerative combinatorics and
  context-free languages.
\newblock In Volker Diekert and Bruno Durand, editors, {\em Proc. of STACS'05},
  pages 18--35, Berlin, Heidelberg, 2005. Springer Berlin Heidelberg.

\bibitem{CadilhacFinkelMcKenzie:2012}
Micha{\"e}l Cadilhac, Alain Finkel, and Pierre McKenzie.
\newblock Unambiguous constrained automata.
\newblock In Hsu-Chun Yen and Oscar~H. Ibarra, editors, {\em Proc. of DLT'12},
  volume 7410 of {\em LNCS}, pages 239--250. Springer Berlin Heidelberg, 2012.

\bibitem{CadilhacMazowieckiPapermanPilipczukSenizergues:ICALP:2020}
Micha{\"e}l Cadilhac, Filip Mazowiecki, Charles Paperman, Micha{\l} Pilipczuk,
  and G{\'e}raud S{\'e}nizergues.
\newblock {On Polynomial Recursive Sequences}.
\newblock In Artur Czumaj, Anuj Dawar, and Emanuela Merelli, editors, {\em
  Proc. of ICALP'20}, volume 168 of {\em LIPIcs}, pages 117:1--117:17,
  Dagstuhl, Germany, 2020. Schloss Dagstuhl--Leibniz-Zentrum f{\"u}r
  Informatik.

\bibitem{Cameron:CUP:2017}
Peter~J. Cameron.
\newblock {\em Notes on Counting: An Introduction to Enumerative
  Combinatorics}.
\newblock Australian Mathematical Society Lecture Series. Cambridge University
  Press, 1 edition, 2017.

\bibitem{CastiglioneMassazza:TCS:2017}
Giusi Castiglione and Paolo Massazza.
\newblock On a class of languages with holonomic generating functions.
\newblock {\em Theoretical Computer Science}, 658:74--84, 2017.

\bibitem{ChengKaminski:CFL:AI98}
Edward Y.~C. Cheng and Michael Kaminski.
\newblock Context-free languages over infinite alphabets.
\newblock {\em Acta Inf.}, 35(3):245--267, 1998.

\bibitem{ChomskySchutzenberger:Algebraic:1963}
N.~Chomsky and M.~P. Sch{\"u}tzenberger.
\newblock The algebraic theory of context-free languages.
\newblock In P.~Braffort and D.~Hirschberg, editors, {\em Computer Programming
  and Formal Systems}, volume~35 of {\em Studies in Logic and the Foundations
  of Mathematics}, pages 118--161. Elsevier, 1963.

\bibitem{Clemente:EPTCS:2020}
Lorenzo Clemente.
\newblock On the complexity of the universality and inclusion problems for
  unambiguous context-free grammars.
\newblock In Laurent Fribourg and Matthias Heizmann, editors, {\em {\rm
  Proceedings 8th International Workshop on} Verification and Program
  Transformation {\rm and 7th Workshop on} Horn Clauses for Verification and
  Synthesis, {\rm Dublin, Ireland, 25-26th April 2020}}, volume 320 of {\em
  EPTCS}, pages 29--43. Open Publishing Association, 2020.
\newblock \href {https://doi.org/10.4204/EPTCS.320.2}
  {\path{doi:10.4204/EPTCS.320.2}}.

\bibitem{ClementeLasota:CSL:2015}
Lorenzo Clemente and Slawomir Lasota.
\newblock Reachability analysis of first-order definable pushdown systems.
\newblock In Stephan Kreutzer, editor, {\em Proc. of CSL'15}, volume~41 of {\em
  LIPIcs}, pages 244--259, Dagstuhl, 2015.

\bibitem{Cohn:1995}
P.~M. Cohn.
\newblock {\em Skew Fields: Theory of General Division Rings}, volume~57 of
  {\em Encyclopedia of Mathematics and its Applications}.
\newblock Cambridge University Press, 1995.

\bibitem{Colcombet:STACS:2012}
Thomas Colcombet.
\newblock {Forms of Determinism for Automata (Invited Talk)}.
\newblock In Christoph D{\"u}rr and Thomas Wilke, editors, {\em Proc. of
  STACS'12}, volume~14 of {\em LIPIcs}, pages 1--23, Dagstuhl, Germany, 2012.
  Schloss Dagstuhl--Leibniz-Zentrum fuer Informatik.

\bibitem{Colcombet:DCFS:2015}
Thomas Colcombet.
\newblock Unambiguity in automata theory.
\newblock In Jeffrey Shallit and Alexander Okhotin, editors, {\em Descriptional
  Complexity of Formal Systems}, pages 3--18, Cham, 2015. Springer
  International Publishing.

\bibitem{CzerwinskiFigueiraHofman:CONCUR:2020}
Wojciech Czerwi{\'n}ski, Diego Figueira, and Piotr Hofman.
\newblock {Universality Problem for Unambiguous VASS}.
\newblock In Igor Konnov and Laura Kov{\'a}cs, editors, {\em Proc. of
  CONCUR'20}, volume 171 of {\em LIPIcs}, pages 36:1--36:15, Dagstuhl, Germany,
  2020. Schloss Dagstuhl--Leibniz-Zentrum f{\"u}r Informatik.

\bibitem{DaviaudJurdzinskiLazicMazowieckiPerezWorrell:ICALP:2018}
Laure Daviaud, Marcin Jurdzinski, Ranko Lazic, Filip Mazowiecki, Guillermo~A.
  P{\'e}rez, and James Worrell.
\newblock {When is Containment Decidable for Probabilistic Automatal}.
\newblock In Ioannis Chatzigiannakis, Christos Kaklamanis, D{\'a}niel Marx, and
  Donald Sannella, editors, {\em Proc. of ICALP'18}, volume 107 of {\em
  LIPIcs}, pages 121:1--121:14, Dagstuhl, Germany, 2018. Schloss
  Dagstuhl--Leibniz-Zentrum fuer Informatik.

\bibitem{DemriLazic:FreezeLTL:ACM09}
St{\'e}phane Demri and Ranko Lazić.
\newblock {LTL with the freeze quantifier and register automata}.
\newblock {\em ACM Trans. Comput. Logic}, 10(3):16:1--16:30, April 2009.

\bibitem{FlajoletGerholdSalvy:EJC:2005}
Philippe Flajolet, Stefan Gerhold, and Bruno Salvy.
\newblock On the non-holonomic character of logarithms, powers, and the nth
  prime function.
\newblock {\em Electr. J. Comb.}, 11(2), 2005.

\bibitem{Gerhold:EJC:2004}
Stefan Gerhold.
\newblock On some non-holonomic sequences.
\newblock {\em Electr. J. Comb.}, 11(1), 2004.

\bibitem{Giesbrecht:1998}
M.~Giesbrecht.
\newblock Factoring in skew-polynomial rings over finite fields.
\newblock {\em Journal of Symbolic Computation}, 26(4):463--486, 1998.
\newblock URL:
  \url{http://www.sciencedirect.com/science/article/pii/S0747717198902243},
  \href {https://doi.org/https://doi.org/10.1006/jsco.1998.0224}
  {\path{doi:https://doi.org/10.1006/jsco.1998.0224}}.

\bibitem{GiesbrechtKim:JA:2013}
Mark Giesbrecht and Myung~Sub Kim.
\newblock {Computing the Hermite form of a matrix of Ore polynomials}.
\newblock {\em Journal of Algebra}, 376:341--362, 2013.

\bibitem{HalavaHarjuHirvensaloKarhumaki:techrep:2005}
Vesa Halava, Tero Harju, Mika Hirvensalo, and Juhani Karhum{\"a}ki.
\newblock Skolem's problem - on the border between decidability and
  undecidability, 2005.

\bibitem{HopcroftMotwaniUllman:2000}
John Hopcroft, Rajeev Motwani, and Jeffrey Ullman.
\newblock {\em Introduction to Automata Theory, Languages, and Computation}.
\newblock Addison-Wesley, 2000.

\bibitem{KaminskiFrancez:TCS:1994}
Michael Kaminski and Nissim Francez.
\newblock Finite-memory automata.
\newblock {\em Theoretical Computer Science}, 134(2):329--363, 1994.

\bibitem{KaminskiZeitlin:IJFCS:2010}
Michael Kaminski and Daniel Zeitlin.
\newblock Finite-memory automata with non-deterministic reassignment.
\newblock {\em International Journal of Foundations of Computer Science},
  21(05):741--760, 2010.

\bibitem{KannanBachem:SIAM:JoC:1979}
Ravindran Kannan and Achim Bachem.
\newblock Polynomial algorithms for computing the {Smith} and {Hermite} normal
  forms of an integer matrix.
\newblock {\em SIAM Journal on Computing}, 8(4):499--507, 1979.

\bibitem{Klazar:JCT:2003}
Martin Klazar.
\newblock Bell numbers, their relatives, and algebraic differential equations.
\newblock {\em Journal of Combinatorial Theory, Series A}, 102(1):63--87, 2003.
\newblock URL:
  \url{http://www.sciencedirect.com/science/article/pii/S0097316503000141},
  \href {https://doi.org/https://doi.org/10.1016/S0097-3165(03)00014-1}
  {\path{doi:https://doi.org/10.1016/S0097-3165(03)00014-1}}.

\bibitem{Kuich:multiplicity:1994}
Werner Kuich.
\newblock On the multiplicity equivalence problem for context-free grammars.
\newblock In {\em Proceedings of the Colloquium in Honor of Arto Salomaa on
  Results and Trends in Theoretical Computer Science}, pages 232---250, Berlin,
  Heidelberg, 1994. Springer-Verlag.

\bibitem{LabahnNeigerZhou:JoC:2017}
George Labahn, Vincent Neiger, and Wei Zhou.
\newblock Fast, deterministic computation of the {Hermite} normal form and
  determinant of a polynomial matrix.
\newblock {\em Journal of Complexity}, 42:44--71, 2017.

\bibitem{Lipshitz:D-finite:JA:1989}
Leonard Lipshitz.
\newblock D-finite power series.
\newblock {\em Journal of Algebra}, 122(2):353--373, 1989.

\bibitem{Macpherson:Survey:DM:2011}
Dugald Macpherson.
\newblock A survey of homogeneous structures.
\newblock {\em Discrete Math.}, 311(15):1599--1634, August 2011.

\bibitem{MottetQuaas:ToCS:2020}
Antoine Mottet and Karin Quaas.
\newblock The containment problem for unambiguous register automata and
  unambiguous timed automata.
\newblock {\em Theory of Computing Systems}, 2020.
\newblock \href {https://doi.org/10.1007/s00224-020-09997-2}
  {\path{doi:10.1007/s00224-020-09997-2}}.

\bibitem{MuldersStorjohann:JSC:2003}
T.~Mulders and A.~Storjohann.
\newblock On lattice reduction for polynomial matrices.
\newblock {\em Journal of Symbolic Computation}, 35(4):377--401, 2003.

\bibitem{MurawskiTzevelekos:JCSS:2017}
A.S. Murawski, S.J. Ramsay, and N.~Tzevelekos.
\newblock Reachability in pushdown register automata.
\newblock {\em Journal of Computer and System Sciences}, 87:58--83, 2017.

\bibitem{NeigerRosenkildeSolomatov:ISAAC:2018}
Vincent Neiger, Johan Rosenkilde, and Grigory Solomatov.
\newblock Computing {Popov} and {Hermite} forms of rectangular polynomial
  matrices.
\newblock In {\em Proc. of ISAAC'18}, pages 295---302, New York, NY, USA, 2018.
  ACM.

\bibitem{NevenSchwentickVianu:TOCL:2004}
Frank Neven, Thomas Schwentick, and Victor Vianu.
\newblock Finite state machines for strings over infinite alphabets.
\newblock {\em ACM Trans. Comput. Logic}, 5(3):403---435, July 2004.

\bibitem{Ore:1931}
Oystein Ore.
\newblock Linear equations in non-commutative fields.
\newblock {\em Annals of Mathematics}, 32(3):463--477, 1931.
\newblock URL: \url{http://www.jstor.org/stable/1968245}.

\bibitem{Ore:1933}
Oystein Ore.
\newblock Theory of non-commutative polynomials.
\newblock {\em Annals of Mathematics}, 34(3):480--508, 1933.
\newblock URL: \url{http://www.jstor.org/stable/1968173}.

\bibitem{Raskin:ICALP:2018}
Mikhail Raskin.
\newblock {A Superpolynomial Lower Bound for the Size of Non-Deterministic
  Complement of an Unambiguous Automaton}.
\newblock In Ioannis Chatzigiannakis, Christos Kaklamanis, D{\'a}niel Marx, and
  Donald Sannella, editors, {\em Proc. of ICALP'18}, volume 107 of {\em
  LIPIcs}, pages 138:1--138:11, Dagstuhl, Germany, 2018. Schloss
  Dagstuhl--Leibniz-Zentrum fuer Informatik.

\bibitem{SalomaaSoittola:Book:PowerSeries:1978}
Arto Salomaa and Marti Soittola.
\newblock {\em Automata-theoretic aspects of formal power series}.
\newblock Texts and Monographs in Computer Science. Springer, 1978.

\bibitem{Schmerl:FM:1978}
James Schmerl.
\newblock A decidable $\aleph_0$-categorical theory with a non-recursive
  {Ryll-Nardzewski} function.
\newblock {\em Fundamenta Mathematicae}, 98(2):121--125, 1978.

\bibitem{Segoufin:CSL06}
Luc Segoufin.
\newblock Automata and logics for words and trees over an infinite alphabet.
\newblock In Zolt{\'a}n {\'E}sik, editor, {\em Computer Science Logic}, volume
  4207 of {\em LNCS}, pages 41--57. Springer Berlin Heidelberg, 2006.

\bibitem{Senizergues:ICALP:1997}
G{\'e}raud S{\'e}nizergues.
\newblock The equivalence problem for deterministic pushdown automata is
  decidable.
\newblock In Pierpaolo Degano, Roberto Gorrieri, and Alberto
  Marchetti-Spaccamela, editors, {\em Proc. of ICALP'97}, pages 671--681,
  Berlin, Heidelberg, 1997. Springer Berlin Heidelberg.

\bibitem{Stanley:EJC:1980}
Richard~P. Stanley.
\newblock Differentiably finite power series.
\newblock {\em European Journal of Combinatorics}, 1(2):175--188, 1980.

\bibitem{Stanley:1986}
Richard~P. Stanley.
\newblock {\em Enumerative Combinatorics}.
\newblock The Wadsworth \& Brooks/Cole Mathematics Series 1. Springer, 1
  edition, 1986.

\bibitem{StearnsHunt:Unambiguous:SFCS81}
R.~Stearns and H.~Hunt.
\newblock On the equivalence and containment problems for unambiguous regular
  expressions, grammars, and automata.
\newblock In {\em Proc. of SFCS'81}, pages 74--81, Washington, DC, USA, 1981.
  IEEE Computer Society.
\newblock URL: \url{http://dx.doi.org/10.1109/SFCS.1981.29}, \href
  {https://doi.org/10.1109/SFCS.1981.29} {\path{doi:10.1109/SFCS.1981.29}}.

\bibitem{StockmeyerMeyer:1973}
L.~J. Stockmeyer and A.~R. Meyer.
\newblock Word problems requiring exponential time (preliminary report).
\newblock In {\em Proc. of STOC'73}, pages 1--9, New York, NY, USA, 1973. ACM.

\bibitem{Tzeng:SIAMJC:1992}
Wen-Guey Tzeng.
\newblock A polynomial-time algorithm for the equivalence of probabilistic
  automata.
\newblock {\em SIAM J. Comput.}, 21(2):216--227, April 1992.

\bibitem{Villard:ISAAC:1996}
G.~Villard.
\newblock Computing {Popov} and {Hermite} forms of polynomial matrices.
\newblock In {\em Proc. of ISAAC'96}, pages 250---258, New York, NY, USA, 1996.
  Association for Computing Machinery.

\bibitem{Tzeng:IPL:1996}
Tzeng Wen-Guey.
\newblock On path equivalence of nondeterministic finite automata.
\newblock {\em Information Processing Letters}, 58(1):43--46, 1996.

\end{thebibliography}

\begin{appendix}
    \label{sec:appendix}
    \section{Additional material for \cref{sec:linrec}}
\label{app:linrec}

\subsection{One-dimensional linear recursive sequences}
\label{sec:one-dimensional:linrec}

Let $f(n) : \Q^\N$ be a one-dimensional sequence.
The shift operator $\shift {} : \Q^\N \to \Q^\N$ 
is defined as $(\shift {} f)(n) = f(n+1)$ for every $n \in \N$.
A one-dimensional sequence $f$ is \emph{linear recursive} (linrec)
if there are auxiliary sequences $f = f_1, f_2, \dots, f_m : \Q^\N$
satisfying a system of equations of the form 
\begin{align}
    \label{eq:linrec:1}
    \left\{
    \begin{array}{rcl}
        \shift {} f_1 &=& p_{1,1} \cdot f_1 + \cdots + p_{1, m} \cdot f_m, \\
        &\vdots& \\
        \shift {} f_m &=& p_{m,1} \cdot f_1 + \cdots + p_{m,m} \cdot f_m,
    \end{array}
    \right.
\end{align}
where the $p_{i, j} \in \Q[n]$ are univariate polynomials.
The \emph{order} of a linrec sequence is the smallest $m$
s.t.~it admits a description as above.
Allowing terms on the r.h.s.~of the form $p \in \Q[n]$
does not increase the expressiveness power since univariate polynomials are already linrec
and thus $p$ could be replaced by introducing an auxiliary variable for it.
If we fix the initial conditions $f_1(0), \dots, f_m(0)$,
then the system above has unique solution,
and we can moreover compute all the values $f_i(n)$'s by unfolding the definition.
Amongst innumerable others,
the \emph{Fibonacci sequence} $\shift {}^2 f = \shift {} f + f$ is linrec (even constant recursive)
since we can introduce an auxiliary sequence $g$ and write
$\shift {} f = f + g$ and $\shift {} g = f$.
An example using non-constant polynomial coefficients
is provided by the number $t(n)$ of \emph{involutions} of $\set{1, \dots, n}$ (a.k.a.~\emph{telephone numbers})
since $\shift {}^2 t = \shift {} t + (n+1) \cdot t$;
by introducing an auxiliary sequence $s(n)$,
we have a linrec system $\shift {} t = t + n \cdot s$ and $\shift {} s = t$.

\subsection{Examples of bidimensional linrec sequences}
\label{sec:linrec:examples}

There is a wealth of examples of linrec sequences.
The power sequence $n^k$ is bidimensional linrec since
for $n, k \geq 1$,
$n^k = n \cdot n^{k-1}$ and the two sections
$0^k$ and $n^0$ are certainly constant after the first element.
The sequence of \emph{binomial coefficients} ${n \choose k}$
is linrec since ${n \choose k} = {n - 1 \choose k - 1} + {n - 1 \choose k}$ for $n, k \geq 1$
and the two sections satisfy ${n \choose 0} = 1$ for $n \geq 0$
and ${0 \choose k} = 0$ for $k \geq 1$. 
The \emph{Stirling numbers of the first kind $s(n, k)$} are linrec since
$s(n, k) = s(n-1, k-1) - (n-1) \cdot s(n-1, k)$ for $n, k \geq 1$
and the two sections $s(n, 0) = s(0, k) = 0$ are constant for $n, k \geq 1$.
Similar recurrences appear for the
Stirling numbers of the second kind $\stirling n k$
(as remarked in the introduction),
the \emph{Eulerian numbers} $A(n, k) = (n-k) \cdot A(n-1, m-1) + (k+1) \cdot A(n-1, m)$
the \emph{triangle numbers} $T(n, k) = k \cdot T(n-1, k-1) + k \cdot T(n-1, k)$,
and many more.

As an additional example, consider the \emph{Bell numbers} $B(n)$,
which count the number of non-empty partitions of a set of $n$ elements.
Notice that $B(n)$ is not linrec,
in fact not even P-recursive \cite{Klazar:JCT:2003,Gerhold:EJC:2004}.
The well-known relationship $B(n) = \sum_{k=0}^n \stirling n k$
suggests to consider the partial sums
$C(n, k) = \sum_{i=0}^{k-1} \stirling n k$.
We have $C(n, 0) = 0$ and $C(n+1, k+1) = \stirling n k + C(n+1, k)$,
thus $C$ is linrec
and $B(n) = C(n+1, n+1)$ is its diagonal (shifted by one).


\subsection{Comparison with other classes of sequences}
\label{sec:linrec:comparison}

\subparagraph{Linrec vs.~C-recursive.}

A sequence $f : \Q^{\N^d}$ is \emph{C-recursive} if it satisfies a recursion as in \eqref{eq:linrec}
where the affine operators $A_{i,j}$
are restricted to be of the form $c_{i, j, 0} + c_{i, j, 1} \shift 1 + c_{i, j, 2} \shift 2$
for some constants $c_{i, j, 0}, c_{i, j, 1}, c_{i, j, 2} \in \Q$.
Thus bidimensional C-recursive sequences are linrec by definition.
Since the asymptotic growth of a 1-dimensional C-recursive sequence $f(n)$ is $O(r^n)$ for some constant $r \in \Q$,
the sequence $n! = n \cdot (n-1)!$ is linrec but not C-recursive,
and thus the inclusion is strict.
An useful fact is that zeroness of C-recursive sequences can be solved in \PTIME \cite{StearnsHunt:Unambiguous:SFCS81,Tzeng:SIAMJC:1992}.
\begin{lemma}
    \label{lem:zeroness:C-rec}
    The zeroness problem for a one-dimensional C-recursive sequence can be solved in \PTIME.
\end{lemma}

\begin{proof}
    It is well-known that a one-dimensional C-recursive sequence $f$ of order $m$
    represented as in \eqref{eq:linrec:1} where the $p_{i, j}$'s are rational numbers in $\Q$,
    can be transformed into a single recurrence
    \begin{align*}
        \shift m f = c_0 \cdot \shift0 f + \cdots + c_{m-1} \cdot \shift{m-1} f,
    \end{align*}
    where $c_0, \cdots, c_{m-1} \in \Q$.
    C.f.~the proof of \cite[Lemma 1]{HalavaHarjuHirvensaloKarhumaki:techrep:2005}
    relying on the Cayley-Hamilton theorem,
    or the more recent proof of \cite[Proposition 1]{CadilhacMazowieckiPapermanPilipczukSenizergues:ICALP:2020}
    relying on a linear independence argument.
    It follows that $f = 0$ if, and only if,
    $f(n) = 0$ for $0 \leq n \leq m - 1$.
    The latter condition can be checked in \PTIME by \cref{lem:linrec:compute}.
\end{proof}

\subparagraph{Linrec vs.~P-recursive.}

In dimension one, linrec sequences are a special case of \emph{P-recursive sequences} \cite{Stanley:EJC:1980}.
The latter class can be defined as those sequences $f : \Q^\N$
satisfying a linear equation of the form
$p_k(n)f(n) + p_{k-1}(n)f(n-1) + \cdots + p_0(n)f(n-k) = 0$ for every $n \geq k$,
where $p_k(n), \dots, p_0(n) \in \Q[n]$.
Thus linrec corresponds to P-recursive with leading polynomial coefficient $p_k(i) = 1$.
The inclusion is strict.
The Catalan numbers $C(n)$ are P-recursive
since they satisfy $(n+2) \cdot C(n+1) = (4n + 2) \cdot C(n)$ for every $n \geq 0$.
However, they are not linrec,
and in fact not even polyrec (a more general class, c.f.~below),
since
\begin{inparaenum}[1)]
    \item by \cite[Theorem 6]{CadilhacMazowieckiPapermanPilipczukSenizergues:ICALP:2020}
    polyrec (and thus linrec) sequences are ultimately periodic modulo every sufficiently large prime, and
    \item $C(n)$ is not ultimately periodic modulo any prime $p$ \cite{AlterKubota:JCT:1973}. 
\end{inparaenum}

In dimension two, linrec and P-recursive sequences \cite{Lipshitz:D-finite:JA:1989} are incomparable.
The sequence $f(m, n) = m^n$ is linrec
since $f(m+1, n+1) = (m+1) \cdot f(m+1, n)$,
$f(m, 0) = 1$, and $f(0, n+1) = 0$.
The diagonal of $f$ is thus $f(n, n) = n^n$.
Since P-recursive sequences are closed under taking diagonals 
\cite[Theorem 3.8]{Lipshitz:D-finite:JA:1989}
and $n^n$ is not P-recursive \cite[Section 1, page 5]{FlajoletGerholdSalvy:EJC:2005},
it follows that $m^n$ is not P-recursive either (as a two-dimensional sequence).
%

\subparagraph{Linrec vs.~polyrec}

A one-dimensional sequence $f : \Q^\N$ is \emph{polynomial recursive} (polyrec)
if it satisfies a system of equations as in \eqref{eq:linrec:1}
where the rhs' are polynomial expressions in $\Q[f_1(n), \dots, f_m(n)]$
\cite[Definition 3]{CadilhacMazowieckiPapermanPilipczukSenizergues:ICALP:2020}%
\footnote{Since polynomial coefficients can already be defined in this formalism,
we would obtain the same class by allowing more general expressions in $\Q[n][f_1(n), \dots, f_m(n)]$.}.
In dimension one, the class of linrec sequences is strictly included in the class of polyrec sequences.
Consider the sequence $f(n) = 2^{2^n}$.
On the one hand,
it is polyrec since $f(n+1) = f(n)^2$.
On the other hand,
it is not linrec,
and in fact not even P-recursive,
since a P-recursive sequence $g(n)$
has growth rate $O((n!)^c)$ for some constant $c \in \N$
\cite[Proposition 3.11]{Lipshitz:D-finite:JA:1989}.
To the best of our knowledge,
polyrec sequences in higher dimension have not been studied yet.

\subsection{Zeroness problem}

Zeroness of one-dimensional C-recursive sequences is decidable
in \NCTWO \cite{Tzeng:SIAMJC:1992} (and thus in polylogarithmic space);
we recalled a simple argument leading to a \PTIME algorithm in \cref{lem:zeroness:C-rec}.
Zeroness of one-dimensional P-recursive sequences is decidable
(c.f.~\cite{CastiglioneMassazza:TCS:2017} and the corrections in \cite[Section 5]{BostanCarayolKoechlinNicaud:ICALP:2020}).
Zeroness of one-dimensional polyrec sequences is decidable,
and in fact the more general zeroness problem for polynomial automata is decidable with non-primitive recursive complexity \cite{BenediktDuffSharadWorrell:PolyAut:LICS:2017}
(polyrec sequences correspond to polynomial automata over a unary alphabet $\Sigma = \set {a}$).

\subsection{Proofs for \cref{sec:linrec}}

\lemLinrecSection*
\begin{proof}
    We prove the lemma for the $L$-section $f^L(n)$ defined as $f(n, L)$.
    Let the auxiliary sequences be $f = f_1, \dots, f_m$ as in \eqref{eq:linrec},
    and fix the initial conditions $f_j(0, \geq 1), f_j(\geq 1, 0), f_j(0, 0) \in \Q$
    for every $1 \leq j \leq m$.
    Let $f^K_j(n)$ be a new variable denoting the $K$-section $f_j(n, K)$,
    for every $1 \leq j \leq m$
    and $0 \leq K \leq L$.
    We show by induction on $K$ that all the $f^K_j$'s are linrec.
    In the base case $K=0$,
    $f^0_j(n)$ is linrec by setting
    $f^0_j(0) = f_j(0, 0) \in \Q$
    and $\shift1 f^0_j(n) = f_j(n+1, 0) = f_j(\geq 1, 0) \in \Q$.
    Notice that, strictly speaking, the latter is not a legal linrec equation
    since constants are allowed only in the base case
    and not in \eqref{eq:linrec:1} (which are linear systems and not affine ones).
    To this end, we introduce an extra variable $g_j(n)$
    and we define $g_j(0) = f_j(\geq 1, 0) \in \Q$,
    and we have the linrec equations
    \begin{align*}
      \shift1 f^0_j(n) &= g_j(n), \\
      \shift1 g_j(n) &= g_j(n).
    \end{align*}
    For the inductive step, we write
    \begin{align*}
        \shift 1 f^{M+1}_j(n)
            &= \shift 1 \shift 2 f_j(n, M) \\
            &= \sum_i (p_{i00}(n,M) + p_{i01}(n,M) \cdot \shift 1 + p_{i11}(n,M) \cdot \shift 2) f_i(n, M) \\
            &= \sum_i \left((p_{i00}(n,M) + p_{i01}(n,M) \cdot \shift 1) f^M_i(n) + p_{i11}(n,M) \cdot f^{M+1}_i(n)\right).
    \end{align*}
    By induction, each $f^M_i$ is one-dimensional linrec,
    and we can thus adjoin their corresponding systems of equations.
    We have introduced $m \cdot (L+1)$ new variables $f_j^L$'s
    and $m$ variables $g_j$'s
    (thus $m + m \cdot (L+1) + m = m \cdot (L+3)$ in total),
    and the same number of additional equations.
    The initial condition for the new variables $f^M_j$
    is $f^M_j(0) = f_j(0, M)$, which can be computed in \PTIME by \cref{lem:linrec:compute}.
    Moreover every polynomial coefficient appears already in the original system,
    but with the second parameter fixed to some $0 \leq M \leq L$.
    Therefore the degree does not increase
    and the height is bounded by $h \cdot L^d$.
\end{proof}

\section{Proofs for \cref{sec:URA}}
\label{app:URA}

\lemReductionToUniversality*

The two reductions in \cref{lem:reduction:to:universality}
are sufficiently generic to be useful also in other contexts.
For instance, in the context of nondeterministic finite automata
they imply that the inclusion problem $\lang A \subseteq \lang B$
with $A$ nondeterministic and $B$ unambiguous
reduces in \PTIME to the universality problem of an unambiguous finite automaton.
Since the latter problem is in \PTIME \cite[Corollary 4.7]{StearnsHunt:Unambiguous:SFCS81},
the inclusion problem is in \PTIME as well.
Notice that we didn't assume that $A$ is unambiguous,
as it is often done in analogous circumstances
\cite{StearnsHunt:Unambiguous:SFCS81},
\cite[Section 5]{BostanCarayolKoechlinNicaud:ICALP:2020}.
A similar reduction has recently been used in the context of inclusion problems between context-free grammars and finite automata \cite[Sec.~3.1]{Clemente:EPTCS:2020}
In the context of register automata,
the results of \cite{MottetQuaas:ToCS:2020}
do not make any unambiguity assumption on $A$.
%
%
%

\begin{proof}
    Consider two register automata $\AA$ and $\BB$ over finite alphabet $\Sigma$
    with transition relations $\goesto {}_\AA$, resp., $\goesto {}_\BB$.
    We assume w.l.o.g.~that they have the same number of registers.
    %
    %
    %
    Regarding the first point,
    consider the new finite alphabet $\Sigma' = {\goesto {}_\AA}$
    which equals exactly the set of transition rules of $\AA$.
    Let $h : \Sigma' \to \Sigma$ be the surjective homomorphism
    allowing us to recover the original letter
    and defined as $h(p \goesto{\sigma, \varphi} q) = \sigma$;
    We extend $h$ to a function $\hat h : (\Sigma' \times \A) \to (\Sigma \times \A)$
    by preserving the data value $\hat h(t, a) = \tuple{h(t), a}$.
    Consider the automaton $\AA'$ obtained from $\AA$
    by replacing every transition rule $t = (p \goesto {\sigma, \varphi}_\AA q)$ of $\AA$
    with $p \goesto {t, \varphi}_{\AA'} q$.
    Since $\AA'$ has the same set of control locations and number of transitions as $\AA$,
    it is clearly of polynomial size.
    Since $\AA$ is without guessing and orbitised,
    $\varphi$ uniquely determines the next register contents given the current configuration and input $\tuple{\sigma, a}$.
    Thus the only source of nondeterminism in $\AA$ resides in the fact that there may be several transitions over the same $\sigma$.
    This nondeterminism is removed in $\AA'$, since $\sigma$ is replaced by the transition $t$ itself.
    Consequently, $\AA'$ is deterministic.

    Consider the automaton $\BB'$ obtained from $\BB$
    by replacing every transition rule $p \goesto {\sigma, \varphi}_\BB q$
    with \emph{all} transitions of the form $p \goesto {t, \varphi}_{\BB'} q$
    s.t.~$h(t) = \sigma$.
    Clearly, $\BB'$ has the same control locations as $\BB$
    and number of transitions $O(\card{\goesto {}_\AA} \cdot \card{\goesto {}_\BB})$.
    Moreover, if $\BB$ is orbitised, then so it is $\BB'$
    Thus $\BB'$ is of polynomial size
    and by definition $\lang {\BB'} = \hat h^{-1}(\lang \BB)$
    and $\lang \BB = \hat h(\lang {\BB'})$.
    The correctness of the reduction follows from the following claims.
    \begin{claim*}
        $\lang \AA \subseteq \lang \BB$ if, and only if,
        $\lang {\AA'} \subseteq \lang {\BB'}$.
    \end{claim*}
    \begin{proof}[Proof of the claim]
        For the ``only if'' direction, assume $\lang \AA \subseteq \lang \BB$ and let $w \in \lang {\AA'}$.
        By the definition of $\AA'$, $\hat h(w) \in \lang \AA$,
        and thus $\hat h(w) \in \lang \BB$ by assumption.
        It follows that $w \in \hat h^{-1}(\lang \BB) = \lang {\BB'}$, as required.
        
        For the ``if'' direction, assume $\lang {\AA'} \subseteq \lang {\BB'}$
        and let $w = \tuple{\sigma_1, a_1} \cdots \tuple{\sigma_n, a_n} \in \lang \AA$.
        Let the corresponding accepting run in $\AA$ be
        \begin{align*}
            \pi = \tuple{p_0, \bar a_0} \goesto {\sigma_1, a_1} \cdots  \goesto {\sigma_n, a_n} \tuple{p_n, \bar a_n}.
        \end{align*}
        induced by the sequence of transitions $t_1 = (p_0 \goesto {\sigma_1, \varphi_1} p_1), \dots, t_n = (p_{n-1} \goesto {\sigma_n, \varphi_n} p_n)$.
        %
        By the definition of $\AA'$, $\rho := \tuple{t_1, a_1} \cdots \tuple{t_n, a_n} \in \lang {\AA'}$,
        and thus $\rho \in \lang {\BB'}$ by assumption.
        By definition of $\BB'$, $w = \hat h(\rho) \in \hat h(\lang {\BB'}) = \lang \BB$, as required.
    \end{proof}
    \begin{claim*}
        If $\BB$ is unambiguous, then so it is $\BB'$.
    \end{claim*}
    \begin{proof}[Proof of the claim]
        If there are two distinct accepting runs in $\BB'$ over the same input word $w \in (\Sigma' \times \A)*$,
        then applying $\hat h$ yields two distinct accepting runs in $\BB$ over $\hat h(w) \in (\Sigma \times \A)^*$.
    \end{proof}
    \begin{claim*}
        If $\BB$ is without guessing, then so it is $\BB'$.
    \end{claim*}
    \begin{proof}[Proof of the claim]
        If there is a reachable transition in $\sem{\BB'}$ of the form $\tuple{p, \bar a} \goesto {t, a} \tuple{q, \bar a'}$
        s.t.~some fresh $a'_i$ occurs in $\bar a'$,
        then the same holds for
        $\tuple{p, \bar a} \goesto {h(t), a} \tuple{q, \bar a'}$ in $\sem{\BB}$.
    \end{proof}
    
    We now show the second point, and we thus assume that $\AA$ is deterministic.
    By pure set-theoretic manipulations, we have
    \begin{align*}
        \lang \AA \subseteq \lang \BB
        \text{ iff }
            \lang \BB \cup \lang \AA^c = (\A \times A)^*
            \text{ iff }
                (\lang \BB \cap \lang \AA) \cup \lang \AA^c = (\A \times A)^*,
    \end{align*}
    where $\lang \AA^c$ denotes $(\A \times A)^* \setminus \lang \AA$.
    It suffices to observe that
    1) $\lang \AA^c$ is recognisable by a deterministic (and thus unambiguous and without guessing) register automaton
    constructible in \PTIME,
    2) $\lang \BB \cap \lang \AA$ is recognisable by an unambiguous and without guessing automaton of polynomial size
    (since $\AA$ is deterministic and $\BB$ unambiguous and without guessing), and
    3) the disjoint union of two unambiguous and without guessing languages is unambiguous and without guessing, and the complexity is again polynomial.
    We thus take as $\CC$ any unambiguous and without guessing automaton of polynomial size
    \st $\lang \CC = (\lang \BB \cap \lang \AA) \cup \lang \AA^c$.
    Finally, if $\AA$ and $\BB$ are orbitised,
    then $\CC$ is also orbitised.
\end{proof}

\section{Proofs for \cref{sec:URA2linrec}}
\label{app:URA2linrec}

\lemRNRorbit*
\begin{proof}
    Let $R_{p, \bar a}(n, k)$ be the set
    whose cardinality is counted by $\RNR {p, \bar a} (n, k)$:
    \begin{align}
        R_{p, \bar a}(n, k) = \setof{\orbit\pi_{\bar a}}{w \in (\Sigma \times \A)^n, \pi \in \runs {C_I} w {p, \bar a}, \width w = k}.
    \end{align}
    Let $\alpha : \A \to \A$ be an automorphism \st $\alpha(\bar a) = \bar b$.
    We claim that there exists a bijective function
    from $R_{p, \bar a}(n, k)$ to $R_{p, \bar b}(n, k)$.
    Consider the function
    $f$ that maps $\bar a$-orbits of runs to $\bar b$-orbits of runs
    defined as
    \begin{align*}
        f(\orbit \pi_{\bar a})
            = \orbit {\alpha(\pi)}_{\alpha (\bar a)}
            = \orbit {\alpha(\pi)}_{\bar b}.
    \end{align*}
    Since runs $\pi \in R_{p, \bar a}(n, k)$ are $\bar a$-supported
    and $f$ preserves the length of the run
    and the width of the data word labelling it,
    $f$ has the right type $f : R_{p, \bar a}(n, k) \to R_{p, \bar b}(n, k)$.
    We claim that $f$ is injective on $R_{p, \bar a}(n, k)$.
    Towards a contradiction, assume $\orbit \pi_{\bar a} \neq \orbit \rho_{\bar a}$
    but $\orbit {\alpha(\pi)}_{\bar b} = \orbit {\alpha(\rho)}_{\bar b}$.
    There exists a  $\bar b$-automorphism $\beta : \A \to \A$
    \st $\beta (\alpha (\pi)) = \alpha (\rho)$.
    Consequently, $\alpha^{-1} (\beta (\alpha (\pi))) = \rho$ maps $\pi$ to $\rho$.
    Moreover, $\alpha^{-1} \beta \alpha$ is an $\bar a$-automorphism
    since
    \begin{align*}
        \alpha^{-1} (\beta (\alpha (\bar a)))
            &= \alpha^{-1} (\beta (\bar b))
                &\text{(def.~of $\alpha$)} \\
            &= \alpha^{-1} (\bar b)
                &\text{($\beta$ is a $\bar b$-automorphism)} \\
            &= \bar a
                &\text{(def.~of $\alpha$)}.
    \end{align*}
    It follows that $\orbit \pi_{\bar a} = \orbit \rho_{\bar a}$,
    which is a contradiction.
    Thus, $f$ is injective.
    By a symmetric argument,
    there exists also an injective function
    $g : R_{p, \bar b}(n, k) \to R_{p, \bar a}(n, k)$.
\end{proof}

\lemURALinrec*
\begin{proof}
    We show that $G_{p, \orbit{\bar a}}(n, k)$
    counts the number of orbits of initial runs
    over words of length $n$ and width $k$
    ending in a configuration in the orbit $\tuple{p, \orbit{\bar a}}$.
    Let $S_{p, \bar a}(n, k)$ be the set of initial runs ending in $\tuple {p, \bar a}$ over words $w$ of length $n$ and width $k$:
    \begin{align}
        S_{p, \bar a}(n, k) = \setof{\pi}{w \in (\Sigma \times \A)^n, \pi \in \runs {C_I} w {p, \bar a}, \width w = k}.
    \end{align}
    We have $R_{p, \bar a}(n, k) = \orbits {\bar a} {S_{p, \bar a}(n, k)} = \setof {\orbit\pi_{\bar a}} {\pi \in S_{p, \bar a}(n, k)}$.
    We observe the following decomposition for $n, k \geq 0$:
    \begin{align*}
        S_{p', \bar a'}(n+1, k+1) =
         &\underbrace{\bigcup_{t=(p, \bar a \goesto {\sigma, a} p', \bar a'), a \in \bar a}
            \setof {\pi \cdot t} {\pi \in S_{p, \bar a}(n, k+1)}}_{\I}\ \cup \\
        &\underbrace{\bigcup_{t=(p, \bar a \goesto {\sigma, a} p', \bar a'), a \not\in \bar a}
            \setof {\pi \cdot t} {\pi \in S_{p, \bar a}(n, k), a \not\in \pi}}_{\II}\ \cup \\
        &\underbrace{\bigcup_{t=(p, \bar a \goesto {\sigma, a} p', \bar a'), a \not\in \bar a}
            \setof {\pi \cdot t} {\pi \in S_{p, \bar a}(n, k+1), a \in \pi}}_{\III},
    \end{align*}
    where the three unions marked by $\I, \II, \III$ are mutually disjoint.
    When we pass to their $\bar a'$-orbits,
    we also get a disjoint union of orbits:
    \begin{align*}
        R_{p', \bar a'}(n+1, k+1) =
            &\bigcup_{t=(p, \bar a \goesto {\sigma, a} p', \bar a'), a \in \bar a}
             {\setof {\orbit{\pi \cdot t}_{\bar a'}} {\pi \in S_{p, \bar a}(n, k+1)}}_{}\ \cup \\
             &\bigcup_{t=(p, \bar a \goesto {\sigma, a} p', \bar a'), a \not\in \bar a}
             {\setof {\orbit{\pi \cdot t}_{\bar a'}} {\pi \in S_{p, \bar a}(n, k), a \not\in \pi}}_{}\ \cup \\
             &\bigcup_{t=(p, \bar a \goesto {\sigma, a} p', \bar a'), a \not\in \bar a}
            {\setof {\orbit{\pi \cdot t}_{\bar a'}} {\pi \in S_{p, \bar a}(n, k+1), a \in \pi}}_{}.
    \end{align*}
    By taking cardinalities on both sides, we get
    \begin{align*}
        \card{R_{p', \bar a'}(n+1, k+1)} =
            &\card{\bigcup_{t=(p, \bar a \goesto {\sigma, a} p', \bar a'), a \in \bar a}
            \underbrace{\setof {\orbit{\pi \cdot t}_{\bar a'}} {\pi \in S_{p, \bar a}(n, k+1)}}_{R^\I_t}}\ + \\
            &\card{\bigcup_{t=(p, \bar a \goesto {\sigma, a} p', \bar a'), a \not\in \bar a}
            \underbrace{\setof {\orbit{\pi \cdot t}_{\bar a'}} {\pi \in S_{p, \bar a}(n, k), a \not\in \pi}}_{R^\II_t}}\ + \\
            &\card{\bigcup_{t=(p, \bar a \goesto {\sigma, a} p', \bar a'), a \not\in \bar a}
            \underbrace{\setof {\orbit{\pi \cdot t}_{\bar a'}} {\pi \in S_{p, \bar a}(n, k+1), a \in \pi}}_{R^\III_t}}.
    \end{align*}
    \begin{claim}
        Fix two transitions
        $t_1 = (p_1, \bar a_1 \goesto {\sigma_1, a_1} p', \bar a')$ and
        $t_2 = (p_2, \bar a_2 \goesto {\sigma_2, a_2} p', \bar a')$.
        If $R^\I_{t_1} \cap R^\I_{t_2} \not= \emptyset$
        then $\orbit {t_1} = \orbit {t_2}$.
    \end{claim}
    \begin{proof}[Proof of the claim]
        Let $\orbit {\pi_1 \cdot t_1}_{\bar a'} = \orbit {\pi_2 \cdot t_2}_{\bar a'}$
        for two runs
        $\pi_1 \in S_{p_1, \bar a_1}(n-1, k)$ and
        $\pi_2 \in S_{p_2, \bar a_2}(n-1, k)$.
        There exists an ($\bar a'$-)automorphism $\alpha$
        s.t.~$\alpha(\pi_1 \cdot t_1) = \pi_2 \cdot t_2$.
        In particular, $\alpha(t_1) = t_2$,
        i.e., $\orbit {t_1} = \orbit {t_2}$ as required.
    \end{proof}
    The claim above implies that the $R^\I_t$'s
    are disjoint for distinct orbits $\orbit t$'s,
    and similarly for $R^\II_t$ and $R^\III_t$.
    We thus obtain the equations
    \begin{align*}
        \card{R_{p', \bar a'}(n+1, k+1)} =
            &\sum_{\orbit{t=(p, \bar a \goesto {\sigma, a} p', \bar a')}:\; a \in \bar a}
            |\underbrace{\setof {\orbit{\pi \cdot t}_{\bar a'}} {\pi \in S_{p, \bar a}(n, k+1)}}_{R^\I_t}|\ + \\
            &\sum_{\orbit{t=(p, \bar a \goesto {\sigma, a} p', \bar a')}:\; a \not\in \bar a}
            |\underbrace{\setof {\orbit{\pi \cdot t}_{\bar a'}} {\pi \in S_{p, \bar a}(n, k), a \not\in \pi}}_{R^\II_t}|\ + \\
            &\sum_{\orbit{t=(p, \bar a \goesto {\sigma, a} p', \bar a')}:\; a \not\in \bar a}
            |\underbrace{\setof {\orbit{\pi \cdot t}_{\bar a'}} {\pi \in S_{p, \bar a}(n, k+1), a \in \pi}}_{R^\III_t}|.
    \end{align*}
    %
    %
    \begin{claim}
        \label{claim:bijection:I}
        The set of orbits
        $R^\I_t$ is in bijection with the set of orbits
        $$R_{p, \bar a}(n, k+1) = \setof{\orbit\pi_{\bar a}}{\pi \in S_{p, \bar a}(n, k+1)}.$$
    \end{claim}
    \begin{proof}[Proof of the claim]
        Indeed, consider the mapping $f : R^\I_t \to R_{p, \bar a}(n, k+1)$
        defined as%
        \begin{align*}
            f(\orbit{\pi \cdot t}_{\bar a'}) = \orbit\pi_{\bar a}
            \quad \text{ with }
                t = (p, \bar a \goesto {\sigma, a} p', \bar a').
        \end{align*}
        First of all $f$ is well-defined as a function:
        Assume $\orbit{\pi_1 \cdot t}_{\bar a'} = \orbit{\pi_2 \cdot t}_{\bar a'}$
        for two paths $\pi_1, \pi_2$ both ending in configuration $\tuple {p, \bar a}$.
        There exists an $\bar a'$-automorphism $\alpha$
        s.t.~$\alpha(\pi_1 \cdot t) = \pi_2 \cdot t$.
        In particular, $\alpha(\pi_1) = \pi_2$
        and since $\pi_1, \pi_2$ end up in the same configuration $\tuple {p, \bar a}$,
        $\alpha (\bar a) = \bar a$.
        Thus $\alpha$ is in fact a $\bar a$-automorphism
        and $\orbit{\pi_1}_{\bar a} = \orbit{\pi_1}_{\bar a}$
        as required.
        Secondly, $f$ is of the right type since
        $\orbit\pi_{\bar a} \in R_{p, \bar a}(n, k+1)$:
        $\pi \cdot t$ is a run over a word $w \cdot a$ of width $k+1$
        and thus $\pi$ is a run over a word $w$ also of width $k+1$
        because $a \in \bar a$, implying $a \in w$ since the automaton is non-guessing.
        We argue that $f$ is a bijection.
        First of all, $f$ is injective:
        If $f(\orbit{\pi_1 \cdot t}_{\bar a'}) = f(\orbit{\pi_2 \cdot t}_{\bar a'})$,
        then by definition of $f$ we have
        $\orbit {\pi_1}_{\bar a} = \orbit {\pi_2}_{\bar a}$.
        There exists an $\bar a$-automorphism $\alpha$
        s.t.~$\alpha(\pi_1) = \pi_2$.
        Since the automaton is without guessing, $\bar a' \subseteq \bar a$,
        and thus $\alpha$ is also an $\bar a'$-automorphism.
        Since $\alpha(t) = t$
        (due to the fact that $a \in \bar a$ and thus $\alpha (a) = a$),
        $\alpha(\pi_1\cdot t) = \pi_2 \cdot t$
        and thus $\orbit{\pi_1 \cdot t}_{\bar a'} = \orbit{\pi_2 \cdot t}_{\bar a'}$ as required.
        
        The mapping $f$ is also surjective.
        Indeed, let $\orbit \pi_{\bar a} \in R_{p, \bar a}(n, k+1)$.
        Thus $\pi$ ends in configuration $\tuple{p, \bar a}$
        and therefore $\pi \cdot t$ is a run.
        Consequently, $\orbit {\pi \cdot t}_{\bar a'} \in R^\I_t$.
        This is enough since, by the definition of $f$,
        $\orbit \pi_{\bar a} = f(\orbit {\pi \cdot t}_{\bar a'})$.
    \end{proof}
    
    \begin{claim}
        \label{claim:bijection:II}
        The set of orbits $R^\II_t$
        is in bijection with the set of orbits
        $$R_{p, \bar a}(n, k) = \setof{\orbit{\pi}_{\bar a}}{\pi \in S_{p, \bar a}(n, k)}.$$
    \end{claim}
    \begin{proof}[Proof of the claim]
        Consider the mapping
        \begin{align*}
            f(\orbit{\pi \cdot t}_{\bar a'}) = \orbit\pi_{\bar a},
                \quad \text{ with }
                    t = (p, \bar a \goesto {\sigma, a} p', \bar a').
        \end{align*}
        First of all, $f$ is well-defined as a function,
        and the argument is as in the previous point.
        Secondly, $f$ has the right type.
        If $\pi \cdot t$ is a run over a word $w \cdot a$ of width $k+1$,
        then $\pi$ is a run over $w$ of width $k$ since $a \not\in w$.
        Thus $f$ is indeed a mapping from $R_\II$ to $R_{p, \bar a}(n, k)$.
        We argue that $f$ is bijective. First of all, $f$ is injective.
        Consider $\bar a'$-orbit of runs
        $\orbit{\pi_1 \cdot t}_{\bar a'}, \orbit{\pi_2 \cdot t}_{\bar a'} \in R_\II$
        with $a \not\in \pi_1 \cup \pi_2$.
        %
        %
        %
        %
        %
        If $f(\orbit{\pi_1 \cdot t}_{\bar a'}) = f(\orbit{\pi_2 \cdot t}_{\bar a'})$,
        then by definition of $f$ we have
        $\orbit {\pi_1}_{\bar a} = \orbit {\pi_2}_{\bar a}$.
        There exists an $\bar a$-automorphism $\alpha$
        s.t.~$\alpha(\pi_1) = \pi_2$.
        Since $a \not\in \pi_1 \cup \pi_2$,
        there is an automorphism $\beta$
        s.t.~$\beta$ agrees with $\alpha$ on every data value in $\pi_1$
        (in particular, $\beta(\pi_1) = \pi_2$ and $\beta(\bar a) = \bar a$),
        and $\beta(a) = a$.
        Since the automaton is without guessing,
        $\bar a' \subseteq \bar a \cup \set a$.
        Thus, $\beta$ is a $\bar a'$-automorphism
        and $\beta(\pi_1 \cdot t) = \beta (\pi_1) \cdot \beta(t) = \pi_2 \cdot t$,
        i.e., $\orbit{\pi_1 \cdot t}_{\bar a'} = \orbit{\pi_2 \cdot t}_{\bar a'}$
        as required.
        The mapping $f$ is surjective by an argument as in the proof of \cref{claim:bijection:I}.
    \end{proof}
    
    \begin{claim}
        \label{claim:bijection:III}
        The set of orbits
        $R^\III_t$ with $k+1 \geq \width{\bar a}$
        is in bijection with $k+1 - \width{\bar a}$ disjoint copies of the set of orbits
        $$R_{p, \bar a}(n, k+1) = \setof{\orbit{\pi}_{\bar a}}{\pi \in S_{p, \bar a}(n, k+1)},$$
        and it is empty if otherwise $k+1 < \width{\bar a}$.
    \end{claim}
    \begin{proof}[Proof of the claim]
        If $k+1 < \width{\bar a}$, then clearly since the automaton is non-guessing
        it could not have stored more distinct data values $\width{\bar a}$ in the register
        than the number of distinct data values $k+1$ in the input,
        and thus $R^\III_t = \emptyset$ in this case.
        In the following, thus assume $k+1 \geq \width{\bar a}$.
        Let $w = a_1 \cdots a_n \in \A^n$
        be the sequence of data values labelling the run $\pi$,
        and consider the non-contiguous subsequence
        $D_\pi = a_{i_1} \cdots a_{i_{k+1 - \width{\bar a}}}$ of $w$
        consisting of the $k+1 - \width{\bar a}$ distinct elements in $w \setminus\bar a$
        in their order of appearance in $w$ (and thus in $\pi$).
        Consider the function $f$ defined as
        \begin{align*}
        f(\orbit{\pi \cdot t}_{\bar a'}) = \tuple{j, \orbit\pi_{\bar a}}
            \quad \text{ with }
                t = (p, \bar a \goesto {\sigma, a} p', \bar a'),
        \end{align*}
        where $a \not\in\bar a$ equals the unique $a_{i_j} \in D_\pi$.
        First of all, $f$ is well-defined as a function:
        Assume $\tuple{\orbit{\pi_1 \cdot t}_{\bar a'}, \tuple{j_1, \orbit{\pi_1}_{\bar a}}}, \tuple{\orbit{\pi_2 \cdot t}_{\bar a'}, \tuple{j_2, \orbit{\pi_2}_{\bar a}}} \in f$
        with $\orbit{\pi_1 \cdot t}_{\bar a'} = \orbit{\pi_2 \cdot t}_{\bar a'}$.
        %
        There is an $\bar a'$-automorphism $\alpha$
        s.t.~$\alpha(\pi_1 \cdot t) = \pi_2 \cdot t$.
        In particular, $\alpha(\pi_1) = \pi_2$
        and $\alpha(t) = t$,
        which also implies $\alpha(a_1) = a_2$.
        From $\alpha(\pi_1) = \pi_2$,
        we even have that $\alpha$ is a $\bar a$-automorphism,
        and thus $\orbit{\pi_1}_{\bar a} = \orbit{\pi_2}_{\bar a}$.
        We now argue that $j_1 = j_2$.
        Assume $a$ appears in position $j_1$ in $D_{\pi_1}$
        and in position $j_2$ in $D_{\pi_2}$.
        Assume by way of contradiction that $j_1 \neq j_2$.
        We have that $\alpha(a) = a$ appears in position $j_1$
        in $\alpha (D_{\pi_1}) = D_{\alpha(\pi_1)} = D_{\pi_2}$,
        i.e., $a$ also appears in position $j_1$ in $D_{\pi_2}$.
        This is a contradiction, since all elements in $D_{\pi_2}$ are distinct.
        Thus $f$ is indeed a mapping from $R_\III$ to
        $\set{1, \dots, k+1 - \width{\bar a}} \times R_{p, \bar a}(n, k+1)$.
        
        We argue that $f$ is bijective.
        First of all, $f$ is injective.
        Consider $\bar a'$-orbit of runs
        $\orbit{\pi_1 \cdot t}_{\bar a'}, \orbit{\pi_2 \cdot t}_{\bar a'} \in R^\III_t$
        with $a \not\in \bar a, a \in \pi_1, a \in \pi_2$.
        Assume $f(\orbit{\pi_1 \cdot t}_{\bar a'}) = f(\orbit{\pi_2 \cdot t}_{\bar a'})$.
        By the definition of $f$, we have
        $\orbit {\pi_1}_{\bar a} = \orbit {\pi_2}_{\bar a}$,
        and $a$ occurs in the same position $j$ in $D_{\pi_1}$, resp., $D_{\pi_2}$.
        %
        %
        Consequently $\alpha(a)$ occurs at position $j$ in $\alpha(D_{\pi_1}) = D_{\alpha(\pi_1)} = D_{\pi_2}$,
        and thus $\alpha(a) = a$.
        There exists an $\bar a$-automorphism $\alpha$
        s.t.~$\alpha(\pi_1) = \pi_2$.
        Since the automaton is without guessing,
        $\bar a' \subseteq \bar a \cup \set {a}$,
        and thus $\alpha$ is even an $\bar a'$-automorphism.
        %
        This means $\orbit{\pi_1}_{\bar a'} = \orbit{\pi_2}_{\bar a'}$
        and $\alpha(t) = t$,
        and thus $\orbit{\pi_1 \cdot t}_{\bar a'} = \orbit{\pi_2 \cdot t}_{\bar a'}$
        as required.
        The mapping $f$ is surjective by an argument analogous as in the proof of \cref{claim:bijection:I}.
    \end{proof}
    Thanks to \cref{claim:bijection:I,claim:bijection:II,claim:bijection:III},
    we obtain the equations
    \begin{align*}
        \card{R_{p', \bar a'}(n+1, k+1)} =
            &\sum_{\orbit{p, \bar a \goesto {\sigma, a} p', \bar a'}:\; a \in \bar a}
                \card{R_{p, \bar a}(n, k+1)}\ + \\
            &\sum_{\orbit{p, \bar a \goesto {\sigma, a} p', \bar a'}:\; a \not\in \bar a}
                \card{R_{p, \bar a}(n, k)}\ + \\
            &\sum_{\orbit{p, \bar a \goesto {\sigma, a} p', \bar a'}:\; a \not\in \bar a}
                \card{\set{1, \dots, k+1 - \width{\bar a}} \times R_{p, \bar a}(n, k+1)}.
    \end{align*}
    By recalling the definition
    $\RNR {p, \bar a} (n+1, k+1) = \card {R_{p, \bar a}(n+1, k+1)}$,
    we obtain, as required,
    \begin{align*}
        \RNR {p', \bar a'} (n+1, k+1) =
            &\sum_{\orbit{p, \bar a \goesto {\sigma, a} p', \bar a'}:\; a \in \bar a}
                \RNR {p, \bar a}(n, k+1)\ + \\
            &\sum_{\orbit{p, \bar a \goesto {\sigma, a} p', \bar a'}:\; a \not\in \bar a}
                (\RNR {p, \bar a}(n, k) +
                    \max \set{k+1 - \width{\bar a}, 0} \cdot \RNR {p, \bar a}(n, k+1)).
        \qedhere
    \end{align*}
\end{proof}

\lemUniversalityToZeroness*
\begin{proof}
    We can effectively enumerate all orbits of transitions $\orbit {p, \bar a \goesto {\sigma, a} p', \bar a'}$ by enumerating all the exponentially many constraints up to logical equivalence \cite[Ch.~4]{Bojanczyk:AtomBook:2019},
    which can be done in \PSPACE
    since this is the complexity of first-order logic over the equality relation.
    Recall that the Bell number $B(n)$ counts the number of non-empty partitions of a set of $n$ elements.
    The system in \cref{fig:linrec:equations} contains
    $\ell \cdot B(d) + 2 = O(\ell \cdot 2^{d \cdot \log d})$
    equations and variables.
\end{proof}

\section{Proofs and additional material for \cref{sec:ore}}
\label{app:ore}

\thmCLM*
\begin{proof}
    We adapt a proof by Giesbrecht given in the case when $R$ is a field,
    for which there even is a \emph{least} common left multiple \cite[Sec.~2]{Giesbrecht:1998} (c.f.~also \cite[Sec.~2]{Ore:1933}).
    We consider the more general case where $R$ is a ring,
    in which case we will not have any minimality guarantee for the common left multiple.
    
    We first prove that $R[\partial; \sigma]$ has pseudo-division.
    Let consider the nonzero skew polynomials
    $$A = a_m \cdot \partial^m + \cdots + a_0 \quad \text {and} \quad
    B = b_n \cdot \partial^n + \cdots + b_0$$
    where $m\geq n$.
    Let $R_0 = A$.
    The leading term of $B$ is $b_n \cdot \partial^n$
    and thus the leading term of $\partial^{m-n} \cdot B$ is
    $\partial^{m-n} \cdot b_n \cdot \partial^n = \sigma^{m-n}(b_n) \cdot \partial^m$.
    Since $R$ is \CLM,
    there are $a_0'$ and $b_0'$ s.t.~$a_0' \cdot a_m = b_0' \cdot \sigma^{m-n}(b_n)$.
    Therefore, $R_1 := a_0' \cdot R_0 - b_0' \cdot \partial^{m-n} \cdot B$
    has degree strictly less than $m_0 := m = \deg {R_0}$.
    We repeat this operation obtaining a sequence of remainders:
    \begin{align*}
        a_0' \cdot R_0 &= b_0' \cdot \partial^{m_0-n} \cdot B + R_1, \\
        a_1' \cdot R_1 &= b_1' \cdot \partial^{m_1-n} \cdot B + R_2, \\
        &\ \ \vdots \\
        a_{k-1}' \cdot R_{k-1} &= b_{k-1}' \cdot \partial^{m_{k-1}-n} \cdot B + R_k, \\
        a_k' \cdot R_k &= b_k' \cdot \partial^{m_k-n} \cdot B + R_{k+1},
    \end{align*}
    where $m_i := \deg{R_i}$,
    $R_{i+1} := a_i' \cdot R_i - b_i' \cdot \partial^{m_i-n}$,
    and the degrees satisfy $m_0 > m_1 > \cdots > m_k > n > m_{k+1}$.
    By defining $a = a_k' a_{k-1}' \cdots a_0' \in R$,
    taking as quotient the skew polynomial
    $$P = b_k' \cdot \partial^{m_k-n} + a_k' b_{k-1}' \cdot \partial^{m_{k-1}-n} + \cdots + a_k' a_{k-1}' \cdots a_1' b_0' \cdot \partial^{m_0-n} \in R[\partial; \sigma]$$
    and as a remainder $Q = R_{k+1}$ we have, as required, $\deg Q < m$ and
    \begin{align*}
        a \cdot A = P \cdot B + Q.
    \end{align*}
    
    We now show that $R[\partial; \sigma]$ has the \CLM property.
    To this end, let $A_1, A_2 \in R[\partial; \sigma]$ with $\deg {A_1} \geq \deg {A_2}$ be given.
    We apply the pseudo-division algorithm above
    to obtain the sequence
    \begin{align*}
        a_1 \cdot A_1 &= Q_1 \cdot A_2 + A_3, \\
        a_2 \cdot A_2 &= Q_2 \cdot A_3 + A_4, \\
        &\ \ \vdots \\
        a_{k-2} \cdot A_{k-2} &= Q_{k-2} \cdot A_{k-1} + A_k, \\
        a_{k-1} \cdot A_{k-1} &= Q_{k-1} \cdot A_k + A_{k+1},
    \end{align*}
    with $a_1,\ldots,a_{k-1}\in R$, $A_{k+1} = 0$,
    and the degrees of the $A_i$'s are strictly decreasing:
    $\deg {A_2} > \deg {A_3} > \cdots > \deg {A_k}$. 
    Consider the following two sequences of skew polynomials
    \begin{align*}
        &S_1 = 1, \quad S_2 = 0, \quad S_i = a_{i-2} \cdot S_{i-2} - Q_{i-2} \cdot S_{i-1}, \textrm{ and } \\
        &T_1 = 0, \quad T_2 = 1, \quad T_i = a_{i-2} \cdot T_{i-2} - Q_{i-2} \cdot T_{i-1}.
    \end{align*}
    It can easily be verified that $S_i \cdot A_1 + T_i \cdot A_2 = A_i$ for every $0 \leq i \leq k+1$: The base cases $i = 0$ and $i = 1$ are clear;
    inductively, we have
    \begin{align*}
        S_i A_1 + T_i A_2
            &=
            (a_{i-2} \cdot S_{i-2} - Q_{i-2} \cdot S_{i-1})A_1 +
            (a_{i-2} \cdot T_{i-2} - Q_{i-2} \cdot T_{i-1})A_2 = \\
            &= a_{i-2} (S_{i-2}A_1 + T_{i-2}A_2) - Q_{i-2}(S_{i-1}A_1 + T_{i-1}A_2) = \\
            &= a_{i-2} A_{i-2} - Q_{i-2} A_{i-1} = A_i.
    \end{align*}
    In particular, at the end $S_{k+1} \cdot A_1 + T_{k+1} \cdot A_2 = 0$,
    as required. 
    
    It remains to check that $S_{k+1}$ is nonzero.
    We show the stronger property that $\deg {S_i} = \deg{A_2} - \deg{A_{i-1}}$
    for every $3 \leq i \leq k+1$.
    The base case $i = 3$ is clear.
    For the inductive step,
    notice that $\deg {Q_{i-2}} = \deg {A_{i-2}} - \deg {A_{i-1}} > 0$.
    Thus $\deg({Q_{i-2}}\cdot S_{i-1}) = \deg {A_{i-2}} - \deg {A_{i-1}} +  \deg{A_2} - \deg{A_{i-2}} = \deg{A_2} - \deg {A_{i-1}}$.
    Moreover, $\deg (a_{i-2} \cdot S_{i-2}) = \deg S_{i-2} = \deg{A_2} - \deg{A_{i-3}} < \deg{A_2} - \deg{A_{i-2}}$.
    Thus, $\deg {S_i} = \deg({Q_{i-2}}\cdot S_{i-1}) = \deg{A_2} - \deg {A_{i-1}}$,
    as required.
\end{proof}

\lemZeroness*
\begin{proof}
    We recall Lagrange's classical bound on the roots of univariate polynomials.
    \begin{theorem}[Lagrange, 1769]
        \label{thm:Lagrange}
        The roots of a complex polynomial
        $p(z) = \sum_{i=0}^d a_i \cdot z^i$ of degree $d$
        are bounded by $1 + \sum_{0\leq i\leq d-1} \frac{|a_i|}{|a_n|}$.
        In particular, the maximal root of a polynomial $p(k) \in \Q[k]$ with integral coefficients
        is at most $1 + d \cdot \max_i|a_i|$.
    \end{theorem}
    %
    By \cref{thm:Lagrange}, the largest root of the leading polynomial coefficient $p_{i^*, j^*}(k)$
    is $\leq 1 + \deg_k p_{i^*, j^*} \cdot \size {p_{i^*, j^*}} < 2 + e \cdot h$
    and similarly the roots of all the leading polynomial coefficients
    of the cancelling relations for the sections $f(0, n), \dots, f(i^*, n)$
    are $< 2 + e \cdot h$.
    In the following, let
    \begin{align*}
        K = 2 + j^* + e \cdot h.
    \end{align*}
    \begin{claim}
        The one-dimensional section $f(n, L) \in \Q^\N$ for a fixed $L \geq 0$
        is identically zero if, and only if,
        $f(0, L) = f(1, L) = \cdots = f(m \cdot (L+3), L) = 0$.
    \end{claim}
    \begin{proof}[Proof of the claim]
        The ``only if'' direction is obvious.
        By \cref{lem:linrec:section},
        for any fixed $L \in \N$
        the 1-dimensional $L$-section $f(n, L)$ is linrec of order $\leq m \cdot (L+3)$.
        In fact, it is C-recursive of the same order
        since the coefficients do not depend on $n$ and are thus constants.
        It follows that if $f(0, L) = f(1, L) = \cdots = f(m \cdot (L+3), L) = 0$,
        then in fact $f(n, L) = 0$ for every $n \in \N$
        (c.f.~the proof of \cref{lem:zeroness:C-rec}).
    \end{proof}
    \begin{claim}
        The one-dimensional section $f(M, k) \in \Q^\N$ for a fixed $0 \leq M \leq i^*$
        is identically zero if, and only if,
        $f(M, 0) = f(M, 1) = \cdots = f(M, d + e \cdot h) = 0$.
    \end{claim}
    \begin{proof}[Proof of the claim]
        The ``only if'' direction is obvious.
        %
        %
        By assumption, $f(M, k)$ admits a cancelling relation \eqref{eq:cancelling:relation:1}
        of $\shift2$-degree $\ell^* \leq d$
        and leading polynomial coefficient $q_{\ell^*}(k)$ of degree $\leq e$ and height $\leq h$.
        By \cref{thm:Lagrange}, the roots of $q_{\ell^*}(k)$ are bounded by $O(e \cdot h)$.
        It follows that if $f(M, 0) = f(M, 1) = \cdots = f(M, d + e \cdot h) = 0$
        then $f(M, n)$ is identically zero.
    \end{proof}
    \begin{claim}
        \label{claim:zeroness}
        $f = 0$ if, and only if,
        all the one-dimensional sections
        $$f(n, 0), \dots, f(n, K), f(0, k), \dots, f(i^*, k) \in \Q^\N$$
        are identically zero.
    \end{claim}
    \begin{proof}[Proof of the claim]
        The ``only if'' direction is obvious.
        For the ``if'' direction,
        assume all the sections above are identically zero
        as one-dimensional sequences.
        By way of contradiction, let $(n, k)$ be the pair of indices
        which is minimal for the lexicographic order s.t.~$f(n, k) \neq 0$.
        By assumption, we necessarily have $n > i^*$ and $k > K$.
        By \eqref{eq:cancelling:relation} we have
        \begin{align*}
            p_{i^*, j^*}(k - j^*) \cdot f(n, k) =
            \sum_{(i, j) <_\text{lex} (i^* , j^*)} p_{i, j}(n - i^*, k - j^*)\cdot f(n - (i^* - i), k - (j^* - k)).
        \end{align*}
        %
        Since $k > K$, $k - j^* > K - j^* \geq 2 + e \cdot h$,
        we have $p_{i^*, j^*}(k - j^*) \neq 0$ 
        since the largest root of $p_{i^*, j^*}$ is $\leq 1 + e \cdot h$.
        Consequently, there exists $(i, j) <_\text{lex} (i^* , j^*)$
        s.t.~$f(n - (i^* - i), k - (j^* - k)) \neq 0$,
        which contradicts the minimality of $(n, k)$.
    \end{proof}
    By putting together the three claims above it follows that $f$ is identically zero
    if, and only if, $f$ is zero on the set of inputs
    \begin{align*}
        \set{0, \dots, m \cdot (K+3)} \times \set {0, \dots, K} \cup
            \set{0, \dots, i^*} \times \set {0, \dots, d + e \cdot h}.
    \end{align*}
    Let $N = 1 + \max \set{m\cdot (K+3), i^*}$
    and $K' = 1 + \max \set{K, d + e \cdot h}$.
    The condition above can be verified
    by computing $O(N \cdot K')$ values for $f(n, k)$,
    each of which can be done in deterministic time $\softO {m \cdot N \cdot K'}$
    thanks to \cref{lem:linrec:compute},
    together yielding $\softO {m \cdot N^2 \cdot (K')^2}$
    which is $\softO {p(m, i^*, j^*, d, e, h)}$ for a suitable polynomial $p$.
\end{proof}

\thmZeronessLinrec*
\begin{proof}
    We interpret the system of equations \eqref{eq:linrec}
    as the following linear system of equations with coefficients $P_{i, j} \in W_2$.
    \begin{align}
        \label{eq:linrec:Weyl}
        \left\{\begin{array}{rcl}
            P_{1, 1} \cdot f_1 + \cdots + P_{1, m} \cdot f_m &=& 0, \\
            &\vdots& \\
            P_{m, 1} \cdot f_1 + \cdots + P_{m, m} \cdot f_m &=& 0.
        \end{array}\right.
    \end{align}
    The idea is to eliminate all variables $f_m, \dots, f_2$ from \eqref{eq:linrec:Weyl}
    until a \CR for $f_1$ remains.
    %
    %
    W.l.o.g.~We show how to remove the last variable $f_m$.
    %
    %
    %
    %
    The skew polynomial coefficients of $f_m$ in equations $1, \dots, m$
    are $P_{1, m}, \dots, P_{m, m} \in W_2$.
    By $m$ applications of \cref{cor:CLM},
    we can find left multipliers $Q_1, \dots, Q_m \in W_2$
    s.t.~$Q_1 \cdot P_{1, m} = Q_2 \cdot P_{2, m} = \cdots = Q_m \cdot P_{m, m}$.
    We obtain the new system not containing $f_m$
    \begin{align*}
        \left\{\begin{array}{rrrcl}
            (Q_1P_{1, 1}-Q_mP_{m, 1}) \cdot f_1 + \cdots&
                    +(Q_1P_{1, m-1}-Q_mP_{m, m-1}) \cdot f_{m-1}& =& 0, \\
            &&&\vdots& \\
            (Q_{m-1}P_{m-1, 1}-Q_mP_{m, 1}) \cdot f_1 + \cdots &
                    +(Q_{m-1}P_{m-1, m-1}-Q_mP_{m, m-1}) \cdot f_{m-1}& =& 0.
        \end{array}\right.
    \end{align*}
    After eliminating all the other variables $f_{m-1}, \dots, f_2$ in the same way,
    we are finally left with an equation $R \cdot f_1 = 0$ with $R \in W_2$.
    Thanks to a linear independence-argument
    that will be presented in \cref{lem:full_rank_linrec},
    the operator $R$ is not zero.
    (Notice that the univariate assumption is not necessary to carry over the elimination procedure
    and obtain a cancelling relation.)
    Notice that the polynomial coefficients in $R$ are univariate polynomials in $\Q[k]$.
    Let $p_{i^*, j^*}(k)$ be leading polynomial coefficient of $R$ when put in the form \eqref{eq:cancelling:relation}.
    By an analogous elimination argument we can find cancelling relations $R_1, \dots, R_{i^*}$
    for each of the one-dimensional sections
    $f^*(0, k), \dots, f(i^*, k) \in \Q^\N$
    (which are effectively one-dimensional linrec sequences by \cref{lem:linrec:section}) respectively.
    We then conclude by \cref{lem:zeroness}.
\end{proof}

The elimination algorithm presented so far
suffices to decide the universality, inclusion, and equivalence problems
for unambiguous register automata without guessing.
\begin{corollary}
    The universality and equivalence problems for unambiguous register automata without guessing are decidable.
    The inclusion problem $\lang A \subseteq \lang B$ for register automata without guessing is decidable when $B$ is unambiguous.
\end{corollary}
Notice that in the inclusion problem $\lang A \subseteq \lang B$
we do not assume that $A$ is unambiguous.
\begin{proof}
    By \cref{lem:reduction:to:universality},
    inclusion and equivalence reduce to universality.
    By \cref{lem:univ-iff-count-URA},
    the universality problem reduces to the zeroness problem of the sequence
    $\RNR{}$ from \cref{fig:linrec:equations},
    which is linrec by its definition and \cref{lem:URA:linrec}.
    Since the polynomial coefficients in \cref{fig:linrec:equations} are univariate,
    we can decide zeroness of $\RNR{}$ by \cref{thm:zeroness:linrec}.
\end{proof}

\subsection{\CLM examples}
\label{app:CLM:examples}

In this section we illustrate the \CLM property with two examples,
the first for $W_1$ and the second for $W_2$.

\begin{example}
    \label{ex:CLM:W1}
    We give an example of application of the \CLM property in $W_1$.
    Consider the two polynomials $F_1 = \shift1^2 -(k+1) \shift1$
    and $F_2 = -\shift1^2 + \shift1$.
    Since $k$ and $\shift1$ commute, $F_2 \cdot F_1 = F_1\cdot F_2$
    and the multipliers have degree $2$.
    The \CLM algorithm finds multipliers of degree 1:
    %
    \begin{align*}
        1 \cdot F_1 &= (-1) \cdot F_2 + F_3
            &&\textrm{with } F_3 = -k\shift1, \\
        k \cdot F_2 &= \shift1 \cdot F_3 + F_4
            &&\textrm{with } F_4 = k\shift1, \\
        1\cdot F_3 &= (-1) \cdot F_4.
    \end{align*}
    We have $s_1 = 1, s_2 = 0, s_3 = 1, s_4 = -\shift1, s_5 = - \shift1 + 1$
    and $t_1 = 0, t_2 = 1, t_3 = 1, t_4 = -k+\shift1, t_5=-k+\shift1+1$.
    We can thus verify that $s_5 \cdot F_1 = -t_5 \cdot F_2$.
\end{example}

\begin{example}
    \label{ex:CLM:W2}
    We give an example of \CLM property in $W_2$.
    Consider the skew polynomials
    $G_1 = (-\shift1^2 + \shift1)\shift2^2$ and
    $G_2 = (\shift1^2-k\shift1)\shift2-\shift1$.
    Since $\shift2G_2 = (\shift1^2-(k+1)\shift1)\shift2^2-\shift1\shift2$,
    thanks to \cref{ex:CLM:W1} we have
    \begin{align*}
        (\shift1 - k - 1) \cdot G_1 &= (-\shift1+1)\shift2\cdot G_2 + G_3,
        &&\textrm{with } G_3 = (-\shift1^2 + \shift1)\shift2,
    \end{align*}
    which gives the first pseudo-division.
    Analogously, since $(-\shift1+1)\cdot(\shift1^2-k\shift1) = (\shift1-k)\cdot(-\shift1^2+\shift1)$, we have the second and third pseudo-divisions
    \begin{align*}
        (-\shift1+1)\cdot G_2 &= (\shift1-k)\cdot G_3 + G_4,
        &&\textrm{with } G_4 = \shift1^2-\shift1,\\
        1 \cdot G_3 &= -\shift2 \cdot G_4.
    \end{align*}
    We thus have $s_1 = 1, s_2 = 0, s_3 = \shift1 - k - 1, s_4 = -(\shift1-k)\cdot(\shift1 - k - 1), s_5 = (\shift1 - k - 1) - \shift2\cdot(\shift1-k)\cdot(\shift1 - k - 1) = (\shift1 - k - 1) - (\shift1-k-1)\cdot(\shift1 - k - 2)\shift2$ and
    $t_1 = 0, t_2 = 1, t_3 = -(-\shift1+1)\shift2, t_4 = (-\shift1+1)+(\shift1-k)\cdot (-\shift1+1)\shift2, t_5 = -(-\shift1+1)\shift2 +\shift2\cdot ((-\shift1+1)+(\shift1-k)\cdot (-\shift1+1)\shift2) = (\shift1-k-1)\cdot (-\shift1+1)\shift2^2$.
    One can check that $s_5 \cdot G_1 = -t_5 \cdot G_2$.
\end{example}

\subsection{\CR examples}
\label{app:CR:examples}

In this section we present detailed examples of \CR.

\begin{example}
    \label{ex:elimination}
    We continue our running \cref{ex:weyl:equations}.
    Recall the starting equations:
    \equationsExWeyl
    In order to eliminate $G_p$, we need to find a common left multiple of $a_0 = \shift1\shift2$ and $b_0 = 1 + (k+1) \shift2$, i.e.,
    we need to find skew polynomials $c, d$ s.t.~$c\cdot a_0 = d \cdot b_0$.
    It can be verified that taking $c = 1 + (k+2)\shift2$ and $d = \shift1\shift2$ fits the bill.
    We thus remove the first equation and left-multiply by $d$ the second and third equations
    (with $S_1 = S$ for simplicity from now on):
    \begin{align*}
        &\begin{array}{rrrrrrr}
            \underbrace{(\shift 1^2\shift 2^2 - (k+1) \shift 1 \shift 2^2 - \shift 1\shift 2)}_{a_1} \cdot G_q&&& = 0, \\
            & + (\shift 1^2 \shift 2^2 - \shift 1 \shift 2^2) \cdot G_r&& = 0, \\
            - \underbrace{\shift 2}_{b_1} \cdot G_q & - (1 + k \shift 2) \cdot G_r& + \shift 1 \shift 2 \cdot G_s& = 0,
        \end{array} \\
        &(\shift 1 \shift 2 - (k+1)\shift 2 - 1) \cdot S = 0, \\
        &G_s - S + G = 0.
    \end{align*}
    We now remove $G_q$. Since its coefficient $b_1 = \shift 2$ in the third equation is already a multiple of its coefficient $a_1 = \shift 1^2\shift 2^2 - (k+1) \shift 1 \shift 2^2 - \shift 1\shift 2$ in the first equation,
    it suffices to remove the first equation and left-multiply the third equation by ``$\shift 1^2\shift 2 - (k+1) \shift 1 \shift 2 - \shift 1$'':
    \begin{align*}
        &\begin{array}{rrrrrr}
            \underbrace{(\shift 1^2 \shift 2^2 - \shift 1 \shift 2^2)}_{a_2} \cdot G_r&& = 0, \\
            - \underbrace{(\shift 1^2\shift 2 - (k+1) \shift 1 \shift 2 - \shift 1)(1 + k \shift 2)}_{b_2} \cdot G_r& + (\shift 1^2\shift 2 - (k+1) \shift 1 \shift 2 - \shift 1)\shift 1 \shift 2 \cdot G_s& = 0,
        \end{array} \\
        &(\shift 1 \shift 2 - (k+1)\shift 2 - 1) \cdot S = 0, \\
        &G_s - S + G = 0.
    \end{align*}
    We now remove $G_r$, and thus we need to find a \CLM of
    $a_2 = \shift 1^2 \shift 2^2 - \shift 1 \shift 2^2 = (\shift 1 - 1)\shift 1 \shift 2^2$ and
    $b_2 = (\shift 1^2\shift 2 - (k+1) \shift 1 \shift 2 - \shift 1)(1 + k \shift 2) = (\shift 1\shift 2 - (k+1) \shift 2 - 1)(1 + k \shift 2)\shift 1$.
    It can be checked that for $d = (\shift 1 - 1)\shift 2^2$ there exists some $c$ (whose exact value is not relevant here) s.t.~$c \cdot a_2 = d \cdot b_2$.
    We can thus remove the first equation and left-multiply the second one by $d$:
    \begin{align*}
        \underbrace{(\shift 1 - 1)\shift 2^2(\shift 1^2\shift 2 - (k+1) \shift 1 \shift 2 - \shift 1)\shift 1 \shift 2}_{a_3} \cdot G_s& = 0, \\
        (\shift 1 \shift 2 - (k+1)\shift 2 - 1) \cdot S &= 0, \\
        G_s - S + G &= 0.
    \end{align*}
    We can now immediately remove $G_s$ by left-multiplying the last equation by its coefficient $a_3$ in the first equation:
    \begin{align*}
        \underbrace{(\shift 1 \shift 2 - (k+1)\shift 2 - 1)}_{b_3} \cdot S &= 0, \\
        \underbrace{(\shift 1 - 1)\shift 2^2(\shift 1^2\shift 2 - (k+1) \shift 1 \shift 2 - \shift 1)\shift 1 \shift 2}_{a_3}\cdot(- S + G) &= 0.
    \end{align*}
    In order to finish it remains to remove $S$.
    The general approach is to find a \CLM of $a_3$ and $b_3$,
    but we would like to avoid performing too many calculations here.
    Since $b_3 \cdot S = 0$,
    we also have $b_3 \shift 1^2 \shift 2 \cdot S = 0$
    (since $\shift 1^2 \shift 2 \cdot S$ is just a shifted version of $S$,
    and since $a_3$ can be written as
    $a_3 = (\shift 1 - 1)\shift 2^2(\shift 1\shift 2 - (k+1) \shift 2 - 1)\shift 1^2 \shift 2 = (\shift 1 - 1)\shift 2^2 \cdot b_3 \cdot \shift 1^2 \shift 2$,
    it follows that $a_3 \cdot S$ = 0 and we immediately have
    \begin{align*}
        \underbrace{(\shift 1 - 1)\shift 2^2(\shift 1\shift 2 - (k+1)\shift 2 - 1)\shift 1^2 \shift 2}_{a_4}\cdot G &= 0.
    \end{align*}
    Since $a_4$ can be expanded to (as a sum of products).
    \begin{align*}
        a_4 &=
            (\shift 1 - 1)\shift 2^2(\shift 1\shift 2 - (k+1)\shift 2 - 1)\shift 1^2 \shift 2 = \\
            &=
            (\shift 1\shift 2 - (k+3) \shift 2 - 1) \shift 1^2 (\shift 1 - 1)\shift 2^3 = \\
            &=
            \shift 1^4\shift 2^4 - (k+3)\shift 1^3\shift 2^4 - \shift 1^3\shift 2^3 - \shift 1^3\shift 2^4 + (k+3)\shift 1^2\shift 2^4 + \shift 1^2\shift 2^3 = \\
            &=
            \shift 1^4\shift 2^4 - (k+4) \shift 1^3\shift 2^4 - \shift 1^3\shift 2^3 + (k+3)\shift 1^2\shift 2^4 + \shift 1^2\shift 2^3 ,
    \end{align*}
    the sought cancelling relation for $G$, obtained by expanding the equation above,
    is
    \begin{align*}
        G(n+4, k+4) =
            &\ (k+4) \cdot G(n+3, k+4) + G(n+3, k+3)\; + \\
            &- (k+3) \cdot G(n+2, k+4) - G(n+2, k+3).
    \end{align*}
\end{example}

\begin{example}
    \label{ex:CR:two-registers}
    We show a \CR example coming from a two-register deterministic automaton.
    There are three control locations $p, q, r$,
    which are all accepting and $p$ is initial.
    When going from $p$ to $q$ the automaton stores the input in its first register $x_1$.
    When going from $q$ to $r$, the automaton checks that the input is different from what is stored in $x_1$ and stores it in $x_2$.
    Then the automaton goes from $r$ to $r$ itself
    by reading an input $y$ different from both registers,
    $x_1' = x_2$ and $x_2' = y$.
    In this way the automaton accepts all words s.t.~any three consecutive data values are pairwise distinct.
    We have the counting equations:
    \begin{align*}
        G_p(n+1, k+1) &= 0, \\
        G_q(n+1, k+1) &= G_p(n, k) + (k+1) \cdot G_q(n, k+1), \\
        G_r(n+1, k+1) &= G_q(n, k) + k\cdot G_q(n, k+1) + G_r(n, k) + (k-1)\cdot G_r(n, k+1) , \\
        G(n, k) &= S(n, k) - G_p(n, k) - G_q(n, k) - G_r(n, k).
    \end{align*}
    We find the following \CR:
    \begin{align}
        \label{eq:CR:two-registers}
        \begin{array}{rl}
            G(n+4, k+3)\ = &(2k+4) \cdot G(n+3, k+3) + 2 \cdot G(n+3, k+2)\ + \\
            &- (k^2+4k+3)\cdot G(n+2, k+3)\ + \\
            &- (2k+3)\cdot G(n+2, k+2)-G(n+2, k+1).
        \end{array}
    \end{align}
\end{example}

In the last example we consider an automaton which is almost universal.

\begin{example}
    \label{ex:URA:universal}
    Consider the following register automaton $\AA$ with one register $x$
    with unary finite alphabet $\card\Sigma = 1$.
    There are four control locations $p, q, r, s$
    of which $p$ is initial and $s$ is final.
    The automaton accepts all words of length $\geq 2$
    by unambiguously guessing whether or not the last two letters are equal.
    The transitions are
    $p \goesto {x = \bot \land x' = \bot} p$,
    $p \goesto {x = \bot \land x' = y} q$,
    $p \goesto {x = \bot \land x' = y} r$,
    $q \goesto {x = y \land x' = x} s$,
    $r \goesto {x \neq y \land x' = x} s$.
    Equations:
    \begin{align*}
        G_p(n+1, k+1) &= G_p(n, k) + (k + 1) \cdot G_p(n, k+1), \\
        G_q(n+1, k+1) &= G_p(n, k) + (k + 1) \cdot G_p(n, k+1) = G_p(n+1, k+1), \\
        G_r(n+1, k+1) &= G_p(n, k) + (k + 1) \cdot G_p(n, k+1) = G_p(n+1, k+1), \\
        G_s(n+1, k+1) &= G_q(n, k+1) + G_r(n, k) + k \cdot G_r(n, k+1) = \\
        &=(k+2)G_p(n, k+1) + G_p(n, k), \\
        G(n, k) &= S(n, k) - G_s(n, k).
    \end{align*}
    %
    %
    %
    %
    %
    %
    %
    We find the following \CR:
    \begin{align}
        \label{eq:CR:universal}
          G(n+3, k+3) = G(n+2, k+2) + (k+3) \cdot G(n+2, k+3).
    \end{align}
    Thanks to the relationship above,
    we manually check that
    $G(2, 0) = G(2, 1) = G(2, 2) = 0$,
    we can conclude that $G(n, k) = 0$ for every $n, k \geq 2$.
    Indeed, the automaton accepts all words of length $\geq 2$.
\end{example}

\section{Hermite forms}
\label{app:hermite}

In this section we present an elimination algorithm based on the computation of the Hermite normal form for matrices of skew polynomials.
An easy but important observation in order to get good bounds
is that the first Weyl algebra
$W_1 = \Q[k][\shift1; \sigma_1]$ from \cref{sec:ore}
is in fact isomorphic to the (commutative) ring of bivariate polynomials
$\Q[k, \shift1]$.
In places where we need to obtain good complexity bounds,
we will use $W_1'$ instead of $W_1$
and $W_2'$ instead of $W_2$, where
\begin{align}
    \label{eq:our:skew:polynomials}
    W_1' = \Q[k, \shift1]
        \quad \text{and} \quad
            W_2' = W_1'[\shift2; \sigma_2] = \Q[k, \shift1] [\shift2; \sigma_2].
\end{align}
A skew polynomial $P \in W_2$ (or $W_2'$) can be written in a unique way as a finite sum
$\sum_{i,j, k} a_{i, j, k} z^i \shift1^j \shift2^k$ with $a_{i, j, k} \in \Q$.
We define $\deg_z P$ as the largest $i$ s.t.~$a_{i,j, k} \not= 0$ for some $j, k$;
$\deg_{\shift1}$ and $\deg_{\shift2}$ are defined similarly.
The \emph{combined degree} $\deg_{\shift1+\shift2} P$ is the largest $j+k$ s.t.~$a_{i,j,k} \not= 0$ for some $i$,
and similarly for $\deg_{z+\shift1}$.
The \emph{height} of $P$ is $\size P = \max_{i, j, k} {\size {a_{i, j, k}}}$.

\subparagraph{Rational skew fields.}

The improved elimination algorithm does not work in the skew polynomial ring,
but in its rational field extension. To this end we need to introduce skew fields.
A \emph{skew field} $\F$ is a field where multiplication is not necessarily commutative \cite{Cohn:1995}.
(Skew fields are sometimes called \emph{division rings} since they are noncommutative rings where multiplicative inverses exist.)
%
%
%
%
%
In the same way as the ring of polynomials $\F[x]$ over a field $\F$
can be extended to a rational polynomial field $\F(x)$,
a skew polynomial ring $\F[\partial;\sigma]$ over a skew field $\F$
can be extended to a \emph{rational skew field} $\F(\partial; \sigma)$.
Its elements are formal fractions $\frac{P}{Q} = Q^{-1}P$
quotiented by $Q^{-1}P\sim S^{-1}R$
if there exist $A, B \in \F[\partial;\sigma]$
s.t.~$A \cdot P = B \cdot R$ and $A \cdot Q = B \cdot S$.
Given $P,Q,R,S\in \F[\partial;\sigma]$
s.t.~
$S_1 \cdot Q = Q_1 \cdot S$ and 
$S_1 \cdot P = P_1 \cdot S$
for some $P_1, S_1, Q_1 \in \F[\partial;\sigma]$,
we can define the operations:
\begin{align*}
    \frac{P}{Q} + \frac{R}{S} = \frac{S_1 \cdot P + Q_1 \cdot R}{S_1 \cdot Q}, \qquad  
    \frac{P}{Q} \cdot \frac{R}{S} = \frac{P_1 R}{S_1 Q}, \qquad
    \left(\frac{P}{Q}\right)^{-1} = \frac{Q}{P}.
\end{align*}
%
%
It was shown by O.~Ore that this yields a well-defined skew field structure to $\F(\partial;\sigma)$
and that unique reduced representations $\frac P Q$ exist \cite{Ore:1931}%
\footnote{Actually, Ore considered formal quotients of the form $PQ^{-1}$,
but we found it more convenient to work in the symmetric definition.}.
In our context, we define the skew fields
\begin{align}
    \label{eq:our:skew:fields}
    \F(W_1') = \Q(k, \shift 1)
        \quad \text{and} \quad
            \F(W_2') = \F(W_1') (\shift 2;\sigma_2) = \Q(k, \shift 1)(\shift 2;\sigma_2)
\end{align}
associated to the corresponding iterated Weyl algebras $W_1'$ and $W_2'$. 
Note that $\F(W_1')$ is in fact just a rational (commutative) field of bivariate polynomials. %
For $R = \frac P Q \in \F(W_1')$ or $\F(W_2')$ written in reduced form,
we define $\size R = \max \set{\size P, \size Q}$.

\subparagraph*{Non-commutative linear algebra.}

Let $\F$ be a skew field.
We denote by $\F^{n\times m}$ the ring of matrices $A$ with $n$ rows and $m$ columns
with entries in $\F$, equipped with the usual matrix operations ``$+$'' and ``$\cdot$''.
The \emph{height} of $A \in\F^{n\times m}$ is $\size A = \max_{i, j} \size {A_{i, j}}$.
The \emph{left $\F$-module} spanned by the rows of $A = (u_1, \dots, u_n)$
is the set of vectors in $\F^n$ of the form
$a_1 \cdot u_1 + \cdots + a_n \cdot v_n$
for some $a_1, \dots, a_n \in \F$.
%
The \emph{rank} of $A$ is the dimension of the left $\F$-module spanned by its rows.
%
In other words, the rank of $A$ is the largest integer $r$
s.t.~we can extract $r$ rows $u_{i_1}, \ldots, u_{i_r}$ that are free:
for every $a_1, \ldots, a_r \in \F$,
$a_1 \cdot u_{i_1} + \cdots + a_k \cdot u_{i_k} = 0 $ implies $a_1 = \cdots = a_k = 0$.
A square matrix $A \in \F^{n\times n}$ is \emph{non-singular}
if there exists a matrix $B$ such that $A \cdot B = I$
, where $I \in \F^{n\times n}$ is the identity matrix.
%

The following lemma implies that matrices arising from linrec systems have full rank.
We used this lemma to justify why the elimination algorithm
in the proof of \cref{thm:zeroness:linrec}
successfully produces a non-zero \CR.
\begin{lemma}
    \label{lem:full_rank_linrec}
    Let $A \in W_2 = \Q[n,k][\shift1;\sigma_1][\shift2;\sigma_2]^{n\times n}$ be a matrix of skew polynomials s.t.~the combined degree $\deg_{\shift1+\shift2} A_{i, i}$ of the diagonal entries
    is strictly larger than the combined degree $\deg_{\shift1+\shift2} A_{j, i}$ of every other entry $j \neq i$ in the same column $i$.
    Then $A$ has rank $n$.
\end{lemma}
Indeed, the combined degree of diagonal entries $\shift1\shift2$
in a system of linrec equations \eqref{eq:linrec} is $2$,
while every other entry has the form $p(n, k)$,
$p(n, k) \cdot \shift1$, or
$p(n, k) \cdot \shift2$ with $p(n, k) \in Q[n, k]$
and thus has combined degree $1$.
\begin{proof}
    We denote by $A_i$ the $i^\text{th}$ row of $A$.
    By contradiction, assume $A$ does not have full rank.
    There exist rows $A_{i_1},\ldots,A_{i_k}$
    and nonzero coefficients $P_1,\cdots,P_k\in W_2$ such that:
    $$ P_{i_1} \cdot A_{i_1} + \cdots + P_{i_k} \cdot A_{i_k} = 0.$$
    Let $j_1=i_1$.
    Since $\deg_{\shift1+\shift2} {A_{i_1,i_1}} > \deg_{\shift1+\shift2} A_{{i_r},{i_1}}$ for $r\geq 2$,
    there is an index $j_2$ such that $\deg_{\shift1+\shift2}{P_{j_2}} > \deg_{\shift1+\shift2}{P_{j_1}}$. By repeating this process, we have a sequence of indices $j_1,\ldots,j_{k+1}$ such that
    $$\deg_{\shift1+\shift2}{P_{j_{k+1}}} > \deg_{\shift1+\shift2}{P_{j_k}} > \cdots > \deg_{\shift1+\shift2}{P_{j_1}}.$$
    This is a contradiction because there are only $k$ different $P_i$'s.
\end{proof}

\subparagraph{Hermite normal forms.}

Let $A \in \F[\partial;\sigma]^{n\times n}$ be a skew polynomial square matrix.
Let $\deg_\partial A = \max_{i, j}\deg_\partial A_{i, j}$.
We say that $A$ is \emph{unimodular}
if it is invertible in $\F(\partial;\sigma)^{n\times n}$
and moreover the inverse matrix $A^{-1}$ has coefficients already in the skew polynomial ring $\F[\partial;\sigma]$.
We say that $A$ of rank $r$ is in \emph{Hermite form} if 
\begin{inparaenum}[a)]
    \item \label{cond:H:a}%
    exactly its first $r$ rows are non-zero,
    and the first (leading) non-zero entry in each row
    satisfies the following conditions:
    \begin{inparaenum}[b.1)]
        \item \label{cond:H:b.1}%
        it is a monic skew polynomial
        (its leading coefficient is $1 \in \F$),
        \item \label{cond:H:b.2}%
        all entries below it are zero, and
        \item \label{cond:H:b.3}%
        all entries above it have strictly lower degree.
    \end{inparaenum}
\end{inparaenum}
(In particular, a matrix in Hermite form is upper triangular.)
The \emph{Hermite normal form} (\HNF) of a skew polynomial matrix
$A$ of full rank $n$
is the (unique) matrix $H\in \F[\partial;\sigma]^{n\times n}$
in Hermite form
which can be obtained by applying a (also unique) unimodular transformation
$U\in \F[\partial;\sigma]^{n\times n}$ as $H = U \cdot A$.
Existence of $U$ (and thus of $H$) has been shown in
\cite[Theorem 2.4]{GiesbrechtKim:JA:2013},
and uniqueness in
\cite[Theorem 2.5]{GiesbrechtKim:JA:2013}.
The Hermite form $H$ yields directly a cancelling relationship \eqref{eq:cancelling:relation}
for the $n$-th linrec variable $f_n$, as we show in the following example.
(By reordering the equations, we can get an analogous relationship for $f_1$.)

\begin{example}
    Consider the following system of linrec equations:
    \begin{align*}
        \left\{\begin{array}{rrr}
            (\shift 1  - 1) \shift 2 \cdot G_r & - \shift 2 \cdot G_s& = 0, \\
            - (k \shift 2 + 1 ) \cdot G_r & + \shift 1 \shift 2 \cdot G_s& = 0.
        \end{array}\right.
    \end{align*}
    In matrix form we have
    \begin{align}
        \label{ex:matrix:A}
        \underbrace{\begin{pmatrix}
            (\shift 1-1)\shift 2    &- \shift 2 \\
            -k\shift2 - 1           &\shift1\shift2
        \end{pmatrix}}_{A \in W_2^{2 \times 2}} \cdot \underbrace{\begin{pmatrix}
            G_r \\ G_s
            \end{pmatrix}}_{x} = 0.
    \end{align}
    The matrix $A$ above is not in Hermite form;
    one reason is that $(\shift1-1)\shift2$ is not monic as a polynomial in $W_2$
    (because its leading coefficient is $\shift1 - 1 \neq 1$);
    another reason is that the entry $-k\shift2 - 1$ below it is nonzero.
    We show in \cref{ex:Hermite:solution}
    that the Hermite form $H = U \cdot A$ of $A$ is
    %
    \begin{align*}
        H = \begin{pmatrix}
            1  
            & (\frac{k}{\shift1 -1} - \shift1)\shift2 \\
            0           &\shift2^2 - \frac{1}{\shift1^2 - \shift1 - (k+1)}\shift2
        \end{pmatrix}.
    \end{align*}
    This allows us to immediately obtain a cancelling relation for the variable $G_s$ corresponding to the last row.
    Going back to our initial matrix equation $A \cdot x = 0$,
    we have $U A x = H x = 0$ where $x = \tuple{G_r\ G_s}^T$,
    yielding
    \begin{align*}
        \left(\shift2^2 - \frac{1}{\shift1^2 - \shift1 - (k+1)}\shift2\right) \cdot G_s = 0.
    \end{align*}
    By clearing out the denominator
    (an ordinary bivariate polynomial from $\Q[k, \shift1]$),
    we obtain
    \begin{align}
        \label{eq:cancelling:relation:example:2}
        ((\shift1^2 - \shift1 - (k+1)) \cdot \shift2^2 - \shift2)\cdot G_s
            = (\shift1^2\shift2^2 - \shift1 \shift2^2 - (k+1)\shift2^2 - \shift2)\cdot G_s
            = 0
    \end{align}
    yielding the sought cancelling relation for $G_s$ not mentioning any other sequence:
    \begin{align*}
        G_s(k+2, n+2) = G_s(n+1, k+2) + (k+1)\cdot G_s(n, k+2) + G_s(n, k+1).
    \end{align*}
\end{example}

In order to bound the complexity of the Hermite form $H$ in our case of interest,
we will use results from \cite{GiesbrechtKim:JA:2013},
instantiated in the special case of Ore shift polynomials.
These results generalise to skew polynomials analogous complexity bounds for the \HNF 
over integer matrices $\Z^{n \times n}$ \cite{KannanBachem:SIAM:JoC:1979}
and integer univariate polynomial matrices $\Z[z]^{n \times n}$ \cite{Villard:ISAAC:1996,MuldersStorjohann:JSC:2003,LabahnNeigerZhou:JoC:2017,NeigerRosenkildeSolomatov:ISAAC:2018}.
\begin{theorem}
    \label{thm:bounds}
    Let $A \in \F[\partial;\sigma]^{n\times n}$ of full rank $n$
    with \HNF~$H = U \cdot A \in \F[\partial;\sigma]^{n\times n}$.
    \begin{enumerate}
    
        \item $\sum_i \deg_\partial H_{i,i} \leq n \cdot \deg_\partial A$
        \cite[Theorem 4.7, point (a)]{GiesbrechtKim:JA:2013}.
        In particular,
        \begin{align}
            \label{eq:outer:deg:bound}
            \deg_\partial H \leq n \cdot \deg_\partial A.
        \end{align}

        \item For $A \in \F[z][\partial;\sigma]^{n\times n}$
        and $H \in \F(z)[\partial;\sigma]^{n\times n}$
        \cite[Theorem 5.6, point (a)]{GiesbrechtKim:JA:2013},
        \begin{align}
            \label{eq:inner:deg:bound}
            \deg_z H = O(n^2 \cdot \deg_z A \cdot \deg_\partial A)
        \end{align}

        \item For $A \in \Z[z][\partial;\sigma]^{n\times n}$
        and $H \in \Q(z)[\partial;\sigma]^{n\times n}$
        we have \cite[Corollary 5.9]{GiesbrechtKim:JA:2013},
        \begin{align}
            \label{eq:coefficient:bound}
            \log \size H = \softO {n^2 \cdot \deg_z A \cdot (\deg_\partial A + \log \size A)}.
        \end{align}
        
    \end{enumerate}
\end{theorem}
We lift the results of \cref{thm:bounds} from univariate polynomial rings $\F[z], \Z[z]$
to the bivariate polynomial rings $\F[k, \shift1], \Z[k, \shift1]$ that we need in our complexity analysis
by noticing that the latter behave like the former if we replace $\deg_z$ with $\deg_{k+\shift1}$.
The formal result that we need is the following.
%
\begin{lemma}
    \label{lem:inv:bivariate}
    Let $A$ be an invertible matrix in $\Z[k,\shift1]^{n\times n}$.
    Then $\deg_{k+\shift1} A^{-1} \leq n \cdot \deg_{k+\shift1} A$ and
    $\log |A^{-1}|_\infty \leq n^2 (1+\log |A|_\infty + \log \deg_{k+\shift1} A)$.
\end{lemma}
\begin{proof}
    By Cramer's formula, every coefficient of $A^{-1}$ is the quotient of the determinant of a submatrix of $A$ and the determinant of $A$.
    By Lipschitz' formula we have
    $\det(A) = \sum_\sigma \text{sign}(\sigma) A_{1,\sigma_1}\cdots A_{n,\sigma_2}$,
    where $\text{sign}(\sigma) \in \set{-1, 1}$
    and $\sigma$ ranges over all permutations of $\set{1,\ldots,n}$.
    Then we can bound the size of the determinant 
\end{proof}

The two bounds in \cref{lem:bounds} below
are obtained from the last two bounds in \cref{thm:bounds}
by inspecting the proofs in \cite{GiesbrechtKim:JA:2013}
and using the the bounds on inversion of matrices of bivariate polynomials from \cref{lem:inv:bivariate}.
\begin{lemma}
    \label{lem:bounds}
    \begin{enumerate}
        \item For $A \in \F[k,\shift1][\partial;\sigma]^{n\times n}$
        and $H \in \F(k,\shift1)[\partial;\sigma]^{n\times n}$,
        \begin{align}
            \deg_{k+\shift1} H = O(n^2 \cdot \deg_{k+\shift1} A \cdot \deg_\partial A)
        \end{align}

        \item For $A \in \Z[k,\shift1][\partial;\sigma]^{n\times n}$
        and $H \in \Q(k,\shift1)[\partial;\sigma]^{n\times n}$
        we have
        \begin{align}
            \log \size H = \softO{n^2 \cdot \deg_{k+\shift1} A \cdot (\deg_\partial A + \log \size A)}.
        \end{align}
    \end{enumerate}
\end{lemma}
Putting everything together,
the bounds from point 1.~of \cref{thm:bounds}
and the two bounds from \cref{lem:bounds}
yield the following corollary.
\begin{corollary}
    \label{cor:bounds}
    Let $A \in (W_2')^{m\times m} = \Q[k, \shift1][\shift2; \sigma_2]^{m \times m}$ of full rank $m$
    with \HNF~$H = U \cdot A \in \Q(k, \shift1)[\shift2;\sigma_2]^{m \times m}$.
    We have:
    \begin{align*}
        \deg_{\shift2} H &\leq n \cdot \deg_{\shift2} A, \\
            \deg_{k+\shift1} H &= O(n^2 \cdot \deg_{k+\shift1} A \cdot \deg_{\shift2} A), \\
                \log \size H &= \softO{m^2 \cdot \deg_{\shift2} A \cdot (\deg_{k+\shift1} A + \log \size A)}.
    \end{align*}
\end{corollary}
Thus, the degrees of the \HNF are polynomially bounded,
and the heights are exponentially bounded.
The bounds from \cref{cor:bounds} yield the complexity upper-bound
on the zeroness problem that we are after.

\lemCRBounds*
\begin{proof}
    Let $f$ be a linrec sequence of order $\leq m$, degree $\leq d$, and height $\leq h$.
    Since $\deg_{\shift2} = \deg_{\shift1} = 1$ in $A$ from linrec,
    thanks to \cref{cor:bounds} the Hermite form $H$ has
    $\deg_{\shift2} H \leq m$,
    $\deg_{k+\shift1} H$ is polynomially bounded
    (and thus $\deg_k H$ and $\deg_{\shift1} H$ as well), 
    and $\size H$ is exponentially bounded.
    %
    %
    Thanks to the fact that the Hermite form is triangular,
    we can immediately extract from $H \cdot x = 0$
    the existence of a cancelling relation \eqref{eq:cancelling:relation} for $f_1$
    where $i^*, j^*$ are polynomially bounded,
    the degree of $p_{i^*, j^*}$ is polynomially bounded,
    and the height of $\size {p_{i^*, j^*}}$ is exponentially bounded.
    
    Moreover, consider the one-dimensional sections
    $f(0, k), \dots, f(i^*, k) \in \Q^\N$.
    By \cref{lem:linrec:section},
    they are linrec of order $\leq m \cdot (i^*+3)$,
    degree $\leq d$,
    and height $\leq h \cdot (i^*)^d$,
    and thus there are associated matrices $A_0, \dots, A_{i^*}$
    of the appropriate dimensions $\leq (m \cdot (i^*+3)) \times (m \cdot (i^*+3))$
    with coefficients in $\Q[k][\shift2; \sigma_2]$.
    The bounds from \cref{cor:bounds}
    can be applied to this case as well
    and we obtain for each $0 \leq i \leq i^*$
    a cancelling relation \eqref{eq:cancelling:relation:1} $R_i$
    with leading polynomial coefficient $q_{i, \ell_i^*}(k)$ where
    $\ell_i^*$ is polynomially bounded,
    its degree in $k$ is polynomially bounded,
    and the height $\size {q_{i, \ell_i^*}}$ is exponentially bounded.
\end{proof}

%



\subsection{Extended example}

We conclude this section with an extended example
showing how to compute the Hermite form of a skew polynomial matrix,
thus illustrating the techniques of Giesbrecht and Kim \cite{GiesbrechtKim:JA:2013}
leading to \cref{thm:bounds}.
We apply the algorithm on our running example.
%
%
%
For $n \in \N$,
denote with $\F[\partial; \sigma]_{n}$ the semiring of skew polynomials of degree at most $n$ with coefficients in the field $\F$.
Let $\phi_n : \F[\partial; \sigma]_{n} \to \F^{n+1}$
be the bijection that associates to a skew polynomial of degree $\leq n$
the vector of its coefficients,
starting from the one of highest degree.
For instance,
$$\phi_5(5\cdot \partial^3 + 4 \cdot \partial^2 + 7) = \tuple{0, 0, 5, 4, 0, 7}.$$
The \emph{$m$-Sylvester matrix} of a skew polynomial $P \in \F[\partial; \sigma]_{n-m}$
of degree $\leq n-m$
is the matrix $S_n^m(P) \in \F^{(m+1)\times (n+1)}$
defined by
\begin{align}
    S_n^m(P) = \begin{pmatrix}
        \phi_n(\partial^{m}P)\\
        \phi_n(\partial^{m-1}P)\\
        \vdots \\
        \phi_n(\partial^0 P)
    \end{pmatrix}.
\end{align}
%
For example, for $P = 5\cdot \partial^3 + 4 \cdot \partial^2 + 7$
we have
\begin{align*}
    S_5^2(P) = \begin{pmatrix}
        \phi_5(\partial^2 P)\\
        \phi_5(\partial^1 P)\\
        \phi_5(\partial^0 P)
    \end{pmatrix} = \begin{pmatrix}
        5 &4 &0 &7 &0 &0\\
        0 &5 &4 &0 &7 &0\\
        0 &0 &5 &4 &0 &7
    \end{pmatrix}.
\end{align*}
The next lemma shows that sufficiently large Sylvester matrices can be used
to express product of polynomials in terms of products of matrices.
This crucial idea allows one to transform
problems on skew polynomials in $\F[\partial, \sigma]$
to linear algebra problems in the underlying field (or just semiring) $\F$.
\begin{lemma}[\protect{c.f.~\cite[Sec.~1, eq.~(1)]{BostanChyzakSalvyLi:ISAAC:2012}}]
    Let $P, Q \in \F[\partial; \sigma]$ and $n,m \in \N$
    s.t.~$\deg P \leq m$
    and $\deg Q \leq n - \deg P$.
    Then,
    \begin{align*}
        \phi_n(Q \cdot P) = \phi_{m}(Q) \cdot S_n^m(P).
    \end{align*}
\end{lemma}
We extend both $\phi_n$ and $S_n^m$ to skew polynomial matrices in $\F[\partial;\sigma]_n^{k\times k}$
by point-wise application
and then merging all the obtained matrices into a single one.
\begin{example}
    For instance, $\phi_2(A)$ with $A \in \Q[k][\shift1;\sigma_1][\shift2;\sigma_2]^{2\times 2}$
    from \eqref{ex:matrix:A} equals
    \begin{align*}
        \phi_2(A) &= \phi_2 \begin{pmatrix}
            (\shift 1-1)\shift 2    &- \shift 2 \\
            -k\shift2 - 1           &\shift1\shift2
        \end{pmatrix} = \begin{pmatrix}
            \phi_2 ((\shift 1-1)\shift 2)    & \phi_2 (-\shift 2) \\
            \phi_2 (-k\shift2 - 1)           & \phi_2 (\shift1\shift2)
        \end{pmatrix}  \\
        &= \begin{pmatrix}
            0 &\shift 1-1 &0 &0 &{-1} &0 \\
            0 &{-k} &{-1} &0 &\shift 1 &0
        \end{pmatrix} \in \Q[k][\shift1;\sigma_1]^{2 \times 6}
    \end{align*}
    and thus $S_2^1(A) \in \Q[k][\shift1;\sigma_1]^{4 \times 6}$ is
    \begin{align*}
        S_2^1(A) &= S_2^1 \begin{pmatrix}
            (\shift 1-1)\shift 2    &- \shift 2 \\
            -k\shift2 - 1           &\shift1\shift2
        \end{pmatrix} = \\
        &= \begin{pmatrix}
            S_2^1 ((\shift 1-1)\shift 2)    & S_2^1 (-\shift 2) \\
            S_2^1 (-k\shift2 - 1)           & S_2^1 (\shift1\shift2)
        \end{pmatrix} = \\
        &= \begin{pmatrix}
            \begin{pmatrix}
                \phi_2 (\shift2 (\shift 1-1)\shift 2) \\
                \phi_2 ((\shift 1-1)\shift 2)
            \end{pmatrix} &
                \begin{pmatrix}
                    \phi_2 (\shift2(-\shift2)) \\
                    \phi_2 (-\shift 2)
                \end{pmatrix} \\
            \begin{pmatrix}
                \phi_2 (\shift2(-k\shift2 - 1)) \\
                \phi_2 (-k\shift2 - 1)
            \end{pmatrix} &
                \begin{pmatrix}
                    \phi_2 (\shift2\shift1\shift2) \\
                    \phi_2 (\shift1\shift2)
                \end{pmatrix}
        \end{pmatrix}=  \\
        &= \begin{pmatrix}
            \begin{pmatrix}
                \tuple{\shift 1-1 \ 0 \ 0} \\
                \tuple{0 \ \shift 1-1 \ 0}
            \end{pmatrix} &
                \begin{pmatrix}
                    \tuple{{-1} \ 0 \ 0} \\
                    \tuple{0 \ {-1} \ 0}
                \end{pmatrix} \\
            \begin{pmatrix}
                \tuple{{-(k+1)} \ {-1}\ 0} \\
                \tuple{0 \ {-k} \ {-1}}
            \end{pmatrix} &
                \begin{pmatrix}
                    \tuple{\shift 1 \ 0 \ 0} \\
                    \tuple{0 \ \shift 1 \ 0}
                \end{pmatrix}
        \end{pmatrix}=  \\
        &= \begin{pmatrix}
                \shift 1-1 &0 &0 &{-1} &0 &0 \\
                0 &\shift 1-1 &0 &0 &{-1} &0 \\
                {-(k+1)} &{-1} &0 &\shift 1 &0 &0 \\
                0 &{-k} &{-1} &0 &\shift 1 &0
        \end{pmatrix}.
    \end{align*}
\end{example}
By definition of the Hermite form, we have that $H = U \cdot A$.
By \eqref{eq:outer:deg:bound} every degree of skew polynomials appearing therein is bounded by $n \cdot \deg A$. Hence setting $\rho = n \cdot \deg A$,
we have the following matrix equation with coefficients in $\F$:
\begin{align*}
    \phi_{\rho+d}(H) = \phi_{\rho}(U) \cdot S^\rho_{\rho+d}(A).
\end{align*}

The \emph{diagonal degree vector} of the Hermite form for $A$
is the unique vector $d$ s.t.~$d_i = \deg H_{i,i} $.
The algorithm will guess such a vector,
and it can detect whether the guess was correct or not.
If it is the right one, then $H$ and $U$ can be computed.

\begin{example}
    The correct diagonal degree vector for our running example is $(0,2)$.
    The Hermite normal form $H = U \cdot A$ of the $2\times2$ matrix $A$
    from our running example has the form
    \begin{align*}
        H = \begin{pmatrix}
            H_{11}  & H_{12} \\
            0       & H_{22}
        \end{pmatrix}, \quad
        U = \begin{pmatrix}
            U_{11}  & U_{12} \\
            U_{21}  & U_{22}
        \end{pmatrix} \in \Q[k][\shift1;\sigma_1][\shift2;\sigma_2]^{2\times 2}
    \end{align*}
    where $H_{11}, H_{22} \in \Q[k][\shift1;\sigma_1][\shift2;\sigma_2]$
    are \emph{monic} skew polynomials of degree respectively $0$ and $2$
    and $H_{12}, U_{11}, U_{12}, U_{21}, U_{22} \in \Q[k][\shift1;\sigma_1][\shift2;\sigma_2]$
    are skew polynomials of degree $1$.
    It follows that
    \begin{align*}
        \phi_2 (H) &= \begin{pmatrix}
        \phi_2 (H_{11})  & \phi_2 (H_{12}) \\
        0       & \phi_2 (H_{22})
        \end{pmatrix} = \\
        &= \begin{pmatrix}
        \phi_2 ( 1)
            & \phi_2 (a_{121} \shift2 + a_{120}) \\
        0
            & \phi_2 (\shift2^2 + a_{221} \shift2 + a_{220})
        \end{pmatrix} = \\
        &= \begin{pmatrix}
            0
                &0
                    &1
                        &0
                            &a_{121}
                                &a_{120} \\
            0 &0 &0 &1
                        &a_{221}
                            &a_{220}
        \end{pmatrix} \in \Q[k][\shift1;\sigma_1]^{2\times 6}.
    \end{align*}
    Similarly,
    \begin{align*}
        \phi_{1}(U) &= \begin{pmatrix}
            \phi_1(U_{11}) & \phi_1(U_{12}) \\
            \phi_1(U_{21}) & \phi_1(U_{22})
        \end{pmatrix} = \begin{pmatrix}
            \phi_1(u_{111}\shift2 + u_{110}) & \phi_1(u_{121}\shift2 + u_{120}) \\
            \phi_1(u_{211}\shift2 + u_{210}) & \phi_1(u_{221}\shift2 + u_{220})
        \end{pmatrix} = \\
        &= \begin{pmatrix}
            u_{111} &u_{110} &u_{121} &u_{120} \\
            u_{211} &u_{210} &u_{221} &u_{220}
        \end{pmatrix} \in \Q[k][\shift1;\sigma_1]^{2\times 4}.
    \end{align*}
    By putting the pieces together,
    we obtain the following matrix equation with entries in $\Q[k][\shift1;\sigma_1]$
    \begin{align*}
        &\underbrace{\begin{pmatrix}
            0
                &0
                    &1
                        &0
                            &a_{121}
                                &a_{120} \\
            0 &0 &0 &1
                        &a_{221}
                            &a_{220}
        \end{pmatrix}}_{\phi_2(H)} = \\
        &\qquad\underbrace{\begin{pmatrix}
            u_{111} &u_{110} &u_{121} &u_{120} \\
            u_{211} &u_{210} &u_{221} &u_{220}
        \end{pmatrix}}_{\phi_1(U)} \cdot
        \underbrace{\begin{pmatrix}
                \shift 1-1 &0 &0 &{-1} &0 &0 \\
                0 &\shift 1-1 &0 &0 &{-1} &0 \\
                {-(k+1)} &{-1} &0 &\shift 1 &0 &0 \\
                0 &{-k} &{-1} &0 &\shift 1 &0
        \end{pmatrix}}_{S^1_2(A)}.
    \end{align*}
\end{example}

It is shown in \cite[Theorem 5.2]{GiesbrechtKim:JA:2013}
that if we guessed the diagonal degree vector right, then we can remove columns from $\phi_{\rho+d}(H)$ corresponding to under-determined entries,
and corresponding columns in $S^\rho_{\rho+d}(A)$,
in order to obtain two matrices $\Tilde{A}$ and $\Tilde{H}$ such that:
\begin{itemize}
    \item $\Tilde{H}$ is only made of $0$'s and $1$'s.
    \item $\Tilde{A}$ is a square matrix.
    \item The matrix equation $T\Tilde{A} = \Tilde{H}$ of unknown $T$
    (of the same dimensions as $\phi_\rho(U)$) has a unique solution.
    In particular, $\Tilde{A}$ has full rank and hence is invertible.
\end{itemize}

\begin{example}
    The reduced system $\tilde H = \phi_1(U) \cdot \tilde A$
    in our running example is obtained by removing columns
    $5, 6$ from $\phi_2(H)$ and correspondingly from $S_2^1(A)$:
    \begin{align*}
        &\underbrace{\begin{pmatrix}
            0 &0&1&0 \\
            0 &0&0&1
        \end{pmatrix}}_{\tilde H} =
        \underbrace{\begin{pmatrix}
            u_{111} &u_{110} &u_{121} &u_{120} \\
            u_{211} &u_{210} &u_{221} &u_{220}
        \end{pmatrix}}_{\phi_1(U)} \cdot
        \underbrace{\begin{pmatrix}
                \shift1 - 1 & 0 & 0&-1 \\
                0&\shift1-1&0&0 \\
                -(k+1)&-1&0&\shift1 \\
                0&-k&-1&0
        \end{pmatrix}}_{\tilde A}.
    \end{align*}
    
\end{example}

Now the obtained $\tilde A$ is invertible.
Hence we can determine $U$ thanks to the equation $\phi_1(U) = \tilde H \tilde A^{-1}$. 

\begin{example}
    \label{ex:Hermite:solution}
    In the example, we obtain
    \begin{align*}
            T = \begin{pmatrix}
                -\frac{k}{\shift1 - 1}
                & -1 \\
                \frac{k+1}{\shift1^2 - \shift1 - (k+1)}\shift2 + \frac{1}{\shift1^2 - \shift1 - (k+1)}
                & \frac{1}{\shift1^2 - \shift1 - (k+1)} \shift2
            \end{pmatrix},
    \end{align*}
    yielding the Hermite form:
    \begin{align}
            \nonumber
            H = T \cdot A &=
            \begin{pmatrix}
                -\frac{k}{\shift1 - 1}
                & -1 \\
                \frac{k+1}{\shift1^2 - \shift1 - (k+1)}\shift2 + \frac{1}{\shift1^2 - \shift1 - (k+1)}
                & \frac{1}{\shift1^2 - \shift1 - (k+1)} \shift2
            \end{pmatrix}
            \cdot 
            \begin{pmatrix}
                (\shift 1-1)\shift 2    &- \shift 2 \\
                -k\shift2 - 1           &\shift1\shift2
            \end{pmatrix} \\
            \label{eq:Hermite:solution}
            &= \begin{pmatrix}
                1  
                & (\frac{k}{\shift1 -1} - \shift1)\shift2 \\
                0           &\shift2^2 - \frac{1}{\shift1^2 - \shift1 - (k+1)}\shift2
            \end{pmatrix}.
    \end{align}
\end{example}
\end{appendix}

\end{document}